\documentclass[journal]{IEEEtran}

\usepackage{setspace}
\usepackage{graphicx}
\usepackage{amssymb}
\usepackage{amsmath,amsfonts,amssymb,euscript,epsfig,psfrag,
enumerate,float,afterpage, subfigure}%

\newtheorem{thm}{Theorem}

\newtheorem{lem}{Lemma}

\usepackage{cite}

\newcommand{\expect}[1]{\mathbb{E}\left\{#1\right\}}
\newcommand{\defequiv}{\mbox{\raisebox{-.3ex}{$\overset{\vartriangle}{=}$}}}

\begin{document}
\title{Network Capacity Region and Minimum Energy Function for a Delay-Tolerant Mobile Ad Hoc Network}
\author{Rahul~Urgaonkar and Michael~J.~Neely \\
\thanks{This work was supported
in part by the DARPA IT-MANET Program under Grant W911NF-07-0028 and
the National Science Foundation (NSF) under Grant OCE 0520324 and Career
Grant CCF-0747525. This work was presented in part at the 4th International
Symposium on Modeling and Optimization in Mobile, Ad Hoc, and Wireless
Networks (WiOpt), Boston, MA, April 3-7, 2006. 

The authors are with the Department of Electrical Engineering, University of
Southern California, Los Angeles, CA 90089 USA (e-mail: urgaonka@usc.edu;
mjneely@usc.edu).}}

\maketitle

\begin{abstract}
We investigate two quantities of interest in a delay-tolerant mobile ad hoc
network: the network capacity region and the minimum energy
function. The network capacity region is defined as 
the set of all input rates that the network can stably support
considering all possible scheduling and routing algorithms. Given
any input rate vector in this region, the minimum energy function
establishes the minimum time average power required to
support it. In this work, we consider a cell-partitioned model of a delay-tolerant mobile ad hoc network
with general Markovian mobility. This simple model incorporates the essential features of locality of
wireless transmissions as well as node mobility and
enables us to \emph{exactly} compute the corresponding network capacity and minimum energy function. Further,
we propose simple schemes that offer performance guarantees that are
arbitrarily close to these bounds at the cost of an increased delay.
\end{abstract}

\begin{keywords}
delay tolerant networks, mobile ad hoc network, capacity region, minimum energy scheduling, queueing analysis 
\end{keywords}

\section{Introduction}
\label{section:intro}

 Two quantities that characterize the performance limits of a mobile ad hoc network are
the network capacity region and the minimum energy function. The
network capacity region is defined as the set of all
input rates that the network can stably support considering all
possible scheduling and routing algorithms that conform to the given
network structure. The minimum energy function is defined as the minimum time average power (summed over all
users) required to stably support a given input rate vector in this
region. Here, by stability we mean that the input rates are such that for all users, the queues do not grow to
infinity and average delays are bounded. In this paper, we exactly compute these quantities for a
specific model of a delay-tolerant mobile ad hoc network.

Asymptotic bounds on the capacity of static wireless
networks and mobile networks are developed by \cite{gupta-kumar},
\cite{grossglauser-tse}. The work in \cite{grossglauser-tse}
shows that for networks with full uniform mobility, if delay constraints are relaxed, a simple $2$-hop
relay algorithm can support throughput that does not vanish as the
number of network nodes $N$ grows large.
Recent work in \cite{garetto-one} generalizes this model and investigates capacity scaling with non-uniform node mobility and 
heterogeneous nodes.
Capacity-delay tradeoffs in mobile ad hoc networks
are considered in \cite{neely-mobile, toumpis-goldsmith, shah, gaurav-sharma, x-lin}. Flow-based
characterization of the network capacity region is presented in several
works (e.g., \cite{neely-thesis}, \cite{padhye}, \cite{kodialam}).

However, little work has been done in computing the \emph{exact}
capacity and energy expressions for these networks. Exceptions
include a closed form expression for the capacity of a fixed grid
network in \cite{mergen-tong}, an expression for the exact information 
theoretic capacity for a single source multicast setting in a wireless erasure
network \cite{dana},
and an expression for the capacity of
a mobile ad hoc network in \cite{neely-mobile}
that uses a cell-partitioned structure. The work in \cite{neely-mobile} quantizes the network geography into a finite number of cells over
which users move, and assumes that a single packet can be transmitted
between users who are currently in the same cell, while no transmission
is possible between users currently in different cells.\footnote{Here a cell represents only
a sub-region of the network.  There are no base stations in the cells (i.e., this
should not be confused with ``cellular networks'' that use base stations.)}

In this work, we extend this model to more
general scenarios allowing adjacent cell communication and different
rate-power combinations. Specifically, we extend the 
simplified cell-partitioned model of \cite{neely-mobile} (which only allows same cell communication and
considers i.i.d. mobility) to treat adjacent cell communication. We
establish exact capacity expressions for general Markovian user
mobility processes (possibly non-uniform), assuming only a
well-defined steady-state location distribution for the users. {Our
analysis shows that, similar to \cite{neely-mobile}, the capacity is only a function of
the steady-state location distribution of the nodes and a $2$-hop relay
algorithm is throughput optimal for this extended model as well.}
Further, our analysis illuminates the optimal decision strategies
and precisely defines the throughput optimal control law for
choosing between same cell and adjacent cell communication. We then
use this insight to design a simple $2$-hop relay algorithm that can
stabilize the network for all input rates within the network
capacity region. We also compute an upper bound on the average delay
under this algorithm. (Sec. \ref{section:capacity})

We next compute the exact expression for the minimum energy required to
stabilize this network, for all input rates within capacity. Our
result demonstrates a piecewise linear structure for the minimum
energy function that corresponds to opportunistically using up
successive transmission modes. Then we present a {greedy} algorithm
whose average energy  can be pushed arbitrarily close to the minimum
energy at the cost of an increased delay. (Sec. \ref{section:energy})

Before proceeding further, we emphasize that the network capacity
and minimum energy function derived in this paper are subject to the
scheduling and routing constraints of our model as described in the next section. 
Specifically, in this work, we {do not} 
consider techniques that ``mix'' packets, such as network coding or
cooperative communication, which can increase the network
capacity and reduce energy costs.
In fact, in Sec. \ref{section:nc_gains}, we present an example scenario that shows how network coding in conjunction with
the wireless broadcast advantage can increase
the capacity for this model. Calculating the network capacity region and the minimum energy function
when these strategies are allowed is an open problem in general in network information theory and is 
beyond the scope of this paper.

\section{Network Model}
\label{section:model}

\begin{figure}
\centering
\includegraphics[width=5cm, height=5cm]{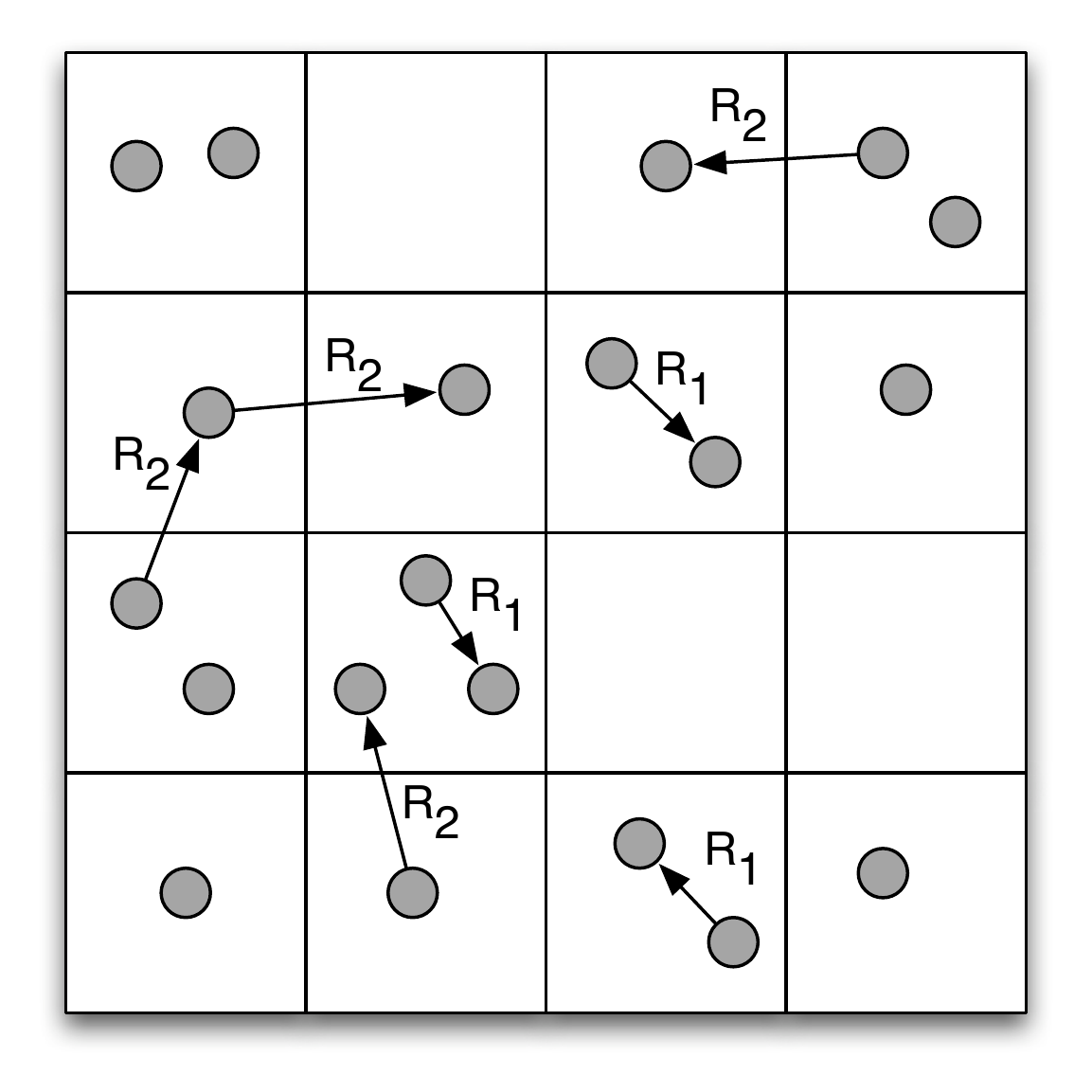}
\caption{An illustration of the cell-partitioned network with same
and adjacent cell communication. Cells
 that share an edge are assumed to be adjacent.}
\label{fig:1}
\end{figure}

\subsection{Cell-Partitioned Structure}
\label{section:cell}
We use a cell-partitioned network model (Fig. \ref{fig:1}) having ${C}$ non-overlapping cells
(not necessarily of the same size/shape). There are ${N}$ users
roaming from cell to cell over the network according to a mobility
process. Each cell ${c} \in \{1, 2, \ldots, {C}\}$ has a set of
adjacent cells $\mathcal{B}_c$ that a user can move into from cell
$c$. The maximum number of adjacent cells of any cell is bounded by a finite constant $J$.
We define the network user density as $\theta = {N}/{C}$
users/cell. For simplicity, ${N}$ is assumed to be even and $N \geq 2$.
Note that there could be ``gaps" in the cell structure due to
infeasible geographic locations. We assume that the gaps do not
partition the network, so that it is possible for a single user to
visit all cells. We assume $C$ is the number of valid cells, not
including such gaps.

\subsection{Mobility Model}
\label{section:mob_model}
Time is slotted so that each user remains in its current cell for a timeslot and potentially moves 
to an adjacent cell at the end of the slot. We assume that each user $i$ moves independently of the 
other users according to a mobility process that is described by a finite state ergodic 
Markov Chain. In particular, let $\mathbf{P}=\{P_{ij}\}_{{C}\times {C}}$ be the transition probability matrix 
of this Markov Chain. Then $P_{ij}$  represents the conditional probability that a user moves to cell 
$j$ in the current slot given that it was in cell $i$ in the last slot. Note that $P_{ij} > 0$ only 
if $j$ is an adjacent cell of $i$, i.e., $j \in \mathcal{B}_i$. It can be shown that the resulting 
mobility process has a well-defined steady-state location distribution $\boldsymbol{\pi} = \{ \pi_c \}_{1 \times {C}}$ 
over the cells $c \in \{1, 2, \ldots, {C}\}$ that satisfies $\boldsymbol{\pi} \mathbf{P} = 
\boldsymbol{\pi}$ and is the same for all users.
However, this distribution could be non-uniform over the cells. We assume that in each 
slot, users are aware of the set of other users in the same cell and in adjacent cells. However, 
the transition probabilities associated with the Markov Chain $\mathbf{P}$ are not necessarily known.

It can be shown (see, for example, \cite{ross}) that the mobility process discussed
above has the following property. Let $\chi(t) \in \{1, \ldots, C\}$ denote the location
of a user in timeslot $t$.  Then, for all integers $d > 0$, there
exist positive constants $\alpha, \gamma$ such that $\forall c \in \{1, 2, \ldots, {C}\}$, the following holds:
\begin{align}
 \pi_c(1 - \alpha \gamma^d) \leq Pr[\chi(t+d) = c | \chi(t)] \leq \pi_c(1 + \alpha \gamma^d)
\label{eq:renewal}
\end{align}
where $\alpha > 1$ and $0 < \gamma < 1$. Moreover, the decay factor $\gamma$ is
given by the second largest eigenvalue of the transition probability
matrix $\mathbf{P}$ (see \cite{seneta}). From this, it can be seen that for any $\epsilon >
0$, choosing $d = \lceil \frac{\log(\epsilon/\alpha)}{\log(\gamma)}\rceil$
ensures that the conditional probability that the user is in cell
$c$ at time $t+ d$ is within $\pi_c \epsilon$ of the steady-state
probability $\pi_c$ of being in cell $c$, irrespective of the
current location. This implies that the Markov Chain converges to
its steady-state probability distribution exponentially fast. 
Using the independence of user mobility processes, the following can be shown
about functionals of the joint user location process $\vec{\chi}(t)$:
\begin{lem}
Let $\vec{\chi}(t) = (\chi_1(t), \ldots, \chi_N(t))$ be the vector of current user
locations, where $\chi_i(t)$ represents the cell of user $i$
in slot $t$. 
Let $f(\vec{\chi}(t))$ be any non-negative function of $\vec{\chi}(t)$, i.e.,
$f(\vec{\chi}(t)) \geq 0 \;\forall \; \vec{\chi}(t)$. 
Define $f_{av}$ as the expectation of $f(\vec{\chi}(t))$ over the steady-state distribution of $\vec{\chi}(t)$:
\begin{align*}
f_{av} \defequiv \sum_{c_1, c_2, \ldots, c_N} f(c_1, \ldots, c_N) \prod_{i=1}^N \pi_{c_i}
\end{align*}
Then for all $d$ such that $\alpha \gamma^d \leq \frac{1}{N^2}$, we have
\begin{align*}
f_{av} (1 - 2N\alpha \gamma^d) \leq \expect{f(\vec{\chi}(t+d))| \vec{\chi}(t)} \leq f_{av} (1 + 2N\alpha \gamma^d) 
\end{align*}
\label{lem:lem1}
\end{lem}
\begin{proof} See Appendix A. 
\end{proof}

\subsection{Traffic Model}
\label{section:tr_model}
We assume that there are $N$ unicast sessions in the network with each node
being the source of one session and the destination of another session. 
Packets are assumed to arrive at the source of each session $i$ according to an i.i.d. arrival process $A_i(t)$ of rate
$\lambda_i$. We assume that in any slot, the maximum number of
arrivals to any session $i$ is bounded, i.e.,  $A_i(t) \leq A_{max}$.
While our analysis holds for the general source-destination pairing,
for simplicity, we assume that $N$ is even with the following one-to-one pairing between users: 
$1 \leftrightarrow 2, 3 \leftrightarrow 4, \ldots, (N - 1)
\leftrightarrow N$, i.e., packets generated by user $1$ are destined
for user $2$ and those generated by user $2$ are destined for user
$1$ and so on. This assumption simplifies the computation of the capacity 
in closed form in Theorem \ref{thm:capacity} and will be used for the rest of the paper.

\subsection{Communication Model}
\label{section:comm_model}
We assume that two users can communicate {only if} they are in the same cell or in adjacent
cells. Further, if the communication takes place in the same cell,
$R_1$ packets can be transmitted from the sender to the receiver if
the sender uses full power. If the receiver is in an adjacent cell,
$R_2$ packets can be transmitted with full power. We assume 
that $R_1$ and $R_2$ are non-negative integers and that $R_1 \geq R_2$. 
Power allocation is restricted to the set $\{0, 1\}$,
i.e., each user either uses zero power or full power. 
For simplicity, we assume that the communication cost consists only of the transmission power. 
The analysis presented can be easily extended to the case with non-zero reception power by defining 
the communication cost as the total power (including transmission and reception) 
required for sending $R_1(R_2)$ packets from a transmitter to a receiver in the same (adjacent) cell. 

We allow \emph{at most} one transmitter in a cell at any given time slot,
though the cell may have multiple receivers (due to possible adjacent cell
communication). Further, a user may potentially transmit and receive
simultaneously. This model is conceivable if the users in
neighboring cells use orthogonal communication channels. This model
allows us to treat scheduling decisions in each cell
\emph{independently} of all other cells, thereby enabling us to
derive closed form expressions for capacity and minimum energy.

\subsection{Discussion of Model}
\label{section:discuss_model}

While an idealization, the cell-partitioned model captures the essential features of locality of
wireless transmissions as well as node mobility and allows us to compute exact
expressions for the network capacity and minimum energy function.
This model is reasonable when nodes use non-interfering orthogonal channels in
adjacent cells. We also refer to Section I-A of \cite{neely-mobile} for further
discussion on the cell-partitioned network assumption.

In this work, we restrict our attention to network control
algorithms that operate according to the network structure
described above. A general algorithm within this class will make
scheduling decisions about what packet to transmit, when, and to
whom. For example, it may decide to transmit to a user in an
adjacent cell rather than to some user in the same cell, even though
the transmission rate is smaller. However, we assume that the
packets themselves are kept intact and are not ``mixed'' (for example,
using cooperative communication and/or network coding). Allowing such strategies
can increase the capacity, although computing the exact capacity region remains a challenging open problem in general.
In Sec. \ref{section:nc_gains}, we present an example that shows how network coding in conjunction with
the wireless broadcast advantage can increase the capacity for this model. However, we note that
if we remove the broadcast feature, then the scenario considered in this paper becomes a network coding problem for multiple unicasts 
over an undirected graph, for which it is not yet known if network coding provides any gains over plain routing
(see further discussion in \cite{netcod_NOW}).

\section{Network Capacity}
\label{section:capacity}
In this section, we compute the exact capacity of the network model
described in the previous section. For simplicity, we assume that all users receive
packets at the same rate
(i.e., $\lambda_i = \lambda$ for all $i$). The capacity of the network is then described
by a scalar quantity which denotes the maximum rate $\lambda$ that
the network can stably support. Recall that network user density $\theta = N/C$ users/cell. Then we have the following:

\begin{thm} The capacity of the network (in packets/slot) is given by:
\begin{displaymath}
\mu = \left\{ \begin{array}{ll}
\frac{R_1q + R_1p + R_2q' + R_2p'}{2\theta} & \textrm{if $R_1 \geq 2R_2$}\\
\frac{2R_1q +  2R_2q'' + R_1p'' + R_2(p'-q')}{2\theta} & \textrm{if
$2R_2
> R_1 \geq R_2$}
\end{array} \right.
\end{displaymath}
where\\
$q = \frac{1}{C} \sum_{c=1}^C Pr[$\textrm{finding a source-destination pair in cell $c$ in a timeslot}] \\
$p = \frac{1}{C} \sum_{c=1}^C Pr[$\textrm{finding at least $2$ users in cell $c$ in a timeslot}] \\
$q' = \frac{1}{C} \sum_{c=1}^C Pr[$\textrm{finding exactly $1$ user in cell $c$ and its destination in an
adjacent cell in a timeslot}] \\
$p' = \frac{1}{C} \sum_{c=1}^C Pr[$\textrm{finding exactly $1$ user in cell $c$ and at least $1$ user in an adjacent cell in a timeslot}] \\
$q'' = \frac{1}{C} \sum_{c=1}^C Pr[$\textrm{finding no source-destination pair in cell $c$ but at least $1$ source-destination pair 
with an adjacent cell in a timeslot}] \\
$p'' = \frac{1}{C} \sum_{c=1}^C Pr[$\textrm{finding no source-destination pair in cell $c$ and any adjacent cell but at least $2$ users in
the cell $c$ in a timeslot}]
\label{thm:capacity}
\end{thm}

The probabilities in the summations above are the probabilities
associated with the steady-state cell location distributions of the
users. Using the independence of user mobility processes and the same 
steady-state cell location distribution $\boldsymbol{\pi} = \{ \pi_c \}_{1 \times {C}}$ for all users, 
we can exactly compute these
probabilities for our model. These are given by (see Appendix B for detailed derivation):
\begin{align}
&q = \frac{1}{C}\sum_{c=1}^C \Big(1 - (1-\pi_c^2)^{\frac{N}{2}}\Big) \nonumber \\
&p = \frac{1}{C}\sum_{c=1}^C \Big(1 - (1-\pi_c)^{N} - N\pi_c(1-\pi_c)^{N-1}\Big) \nonumber\\
&q' = \frac{1}{C}\sum_{c=1}^C \Big(\Pi_{adj}(c)N\pi_c (1-\pi_c)^{N-1}\Big) \nonumber \\
&p' = \frac{1}{C}\sum_{c=1}^C \Big(1 - (1 - \Pi_{adj}(c))^{N-1}\Big)N\pi_c (1 - \pi_c)^{N-1} \nonumber\\
&q'' = \frac{1}{C} \sum_{c=1}^C \sum_{i=1}^{\frac{N}{2}} 2^i\binom{\frac{N}{2}}{i} \pi_c^i (1 - \pi_c)^{N-i} \Big(1 - (1 - \Pi_{adj}(c))^{i}\Big) \nonumber\\ 
&p'' = \frac{1}{C} \sum_{c=1}^C \sum_{i=2}^{\frac{N}{2}} 2^i \binom{\frac{N}{2}}{i} \pi_c^i (1 - \pi_c)^{N-i} (1 - \Pi_{adj}(c))^{i}
\label{eq:prob_values}
\end{align}
Here, $\Pi_{adj}(c)$ denotes the sum of the conditional steady-state
probability of a user being in any adjacent cell of cell $c$ given
that this user is not in cell $c$, i.e., $\Pi_{adj}(c) = \frac{1}{1
- \pi_c}\sum_{i\in \mathcal{B}_c} \pi_i$. Thus, the network can stably support users simultaneously
communicating at any rate $\lambda < \mu$. We prove the theorem in
two parts. First, we establish the necessary condition by deriving
an upper bound on the capacity of \emph{any} stabilizing algorithm.
Then, we establish sufficiency by presenting a specific scheduling
strategy and showing that the average delay is bounded under that
strategy.

\subsection{Proof of Necessity} 
\label{section:capacity_proof}

\begin{proof}
Let $\Psi$ be the set of all stabilizing scheduling policies. Consider any particular
policy $\psi \in \Psi$. Suppose it successfully delivers
$X^\psi_{ab}(T)$ packets from sources to destinations involving ``$a$''
same cell transmissions and ``$b$'' adjacent cell transmissions in
the interval $(0, T )$.
 Fix $\epsilon > 0$. For stability, there must exist arbitrarily large values of
$T$ such that the total output rate is within $\epsilon$ of total
input rate. Thus:
\begin{align}
\frac{\sum_{a=0}^\infty \sum_{b=0}^\infty X^\psi_{ab}(T)}{T} \geq N\lambda - \epsilon 
\label{eq:one}
\end{align}
Define $Y^\psi(T)$ as the total number of packet transmissions 
in $(0, T)$ under policy $\psi$. Then, $Y^\psi(T)$ is \emph{at least} 
$\sum_{a=0}^{\infty} \sum_{b=0}^{\infty} (a + b)X^\psi_{ab}(T)$ 
(because these many packets were certainly delivered). Thus, we have
\begin{align*}
\frac{1}{T} Y^\psi(T) &\geq \frac{1}{T} \sum_{a=0}^{\infty} \sum_{b=0}^{\infty} (a+b)X_{ab}^\psi(T) 
\\ &\geq \frac{1}{T} \sum_{a+b < 2} X_{ab}^\psi(T) + \frac{2}{T} \sum_{a+b \geq 2} X_{ab}^\psi(T)  \\
&\geq \frac{1}{T} \sum_{a+b < 2} X_{ab}^\psi(T) +  2(N\lambda - \epsilon) - \frac{2}{T} \sum_{a+b < 2} X_{ab}^\psi(T) 
\end{align*}
where the last inequality is obtained using (\ref{eq:one}). Hence,
noting that $\epsilon$ can be chosen to be arbitrarily small, we
have:
\begin{align}
\lambda \leq \lim_{T\rightarrow\infty} \frac{Y^\psi(T) + X^\psi_{10}(T) + X_{01}^\psi(T)}{2TN} 
\label{eq:two}
\end{align}
Define $Y_c^\psi(\tau)$ as the total number of packet transmissions
in cell $c$ at timeslot $\tau$ under policy $\psi$. Also
define $X_{10,c}^\psi(\tau)$ and $X_{01,c}^\psi(\tau)$ as the number of
packets delivered by same cell direct  and adjacent cell direct
transmission respectively in cell $c$ at timeslot $\tau$. Then
$Y^\psi(T) + X^\psi_{10}(T) + X^\psi_{01}(T)$ can be written as a sum over
all timeslots $\tau \in (0, T)$ and all cells $c$ as follows:
\begin{align}
&Y^\psi(T) + X^\psi_{10}(T) + X^\psi_{01}(T) \nonumber \\
&= \sum_{\tau=0}^{T-1} \sum_{c=1}^{C} \Big( Y^\psi_c(\tau) + X_{10,c}^\psi(\tau) + X_{01,c}^\psi(\tau) \Big) \nonumber\\
&\leq \sum_{\tau=0}^{T-1} \sum_{c=1}^{C} \max_{\omega \in \Psi} \Big( \hat{Y}^\omega_c(\tau) + \hat{X}_{10,c}^\omega(\tau) +
\hat{X}_{01,c}^\omega(\tau) \Big) \label{eq:three}
\end{align}
where $\hat{Y}^\omega_c(\tau)$ denotes the total number of packet transmission \emph{opportunities}
in cell $c$ at timeslot $\tau$ under any policy $\omega$. Similarly, $\hat{X}_{10,c}^\omega(\tau)$ and $\hat{X}_{01,c}^\omega(\tau)$
denote the total number of packet transmission \emph{opportunities} that use same cell direct and adjacent cell direct
transmissions respectively in cell $c$ at timeslot $\tau$. Note that these do not depend on the queue backlogs and therefore
can be different from the actual number of packet transmissions  (for example, when enough packets are not available). 

 Let $\hat{Z}^\omega_c(\tau) = \hat{Y}^\omega_c(\tau) + \hat{X}^\omega_{10,c}(\tau) +
\hat{X}^\omega_{01,c}(\tau)$. Also define the following indicator decision
variables for any policy $\omega$  for some $\tau \in (0, T)$ and $c \in
\{1, 2, \ldots, {C}\}$:
\begin{displaymath}
I_c^1(\tau) = \left\{ \begin{array}{ll} 1 & \textrm{if a same cell direct transmission can}\\
& \textrm{be scheduled in cell $c$ in slot $\tau$}\\
0 & \textrm{else}
\end{array} \right. 
\end{displaymath}
\begin{displaymath}
I_c^2(\tau) = \left\{ \begin{array}{ll} 1 & \textrm{if a same cell relay transmission can}\\
& \textrm{be scheduled in cell $c$ in slot $\tau$}\\
0 & \textrm{else}
\end{array} \right. 
\end{displaymath}
\begin{displaymath}
I_c^3(\tau) = \left\{ \begin{array}{ll} 1 & \textrm{if an adjacent cell direct transmission can}\\
& \textrm{be scheduled in cell $c$ in slot $\tau$}\\
0 & \textrm{else}
\end{array} \right. 
\end{displaymath}
\begin{displaymath}
I_c^4(\tau) = \left\{ \begin{array}{ll} 1 & \textrm{if an adjacent cell relay transmission can}\\
& \textrm{be scheduled in cell $c$ in slot $\tau$}\\
0 & \textrm{else}
\end{array} \right.
\end{displaymath}
Note that the transmission rates associated with these decision variables are $R_1, R_1, R_2$ and $R_2$ respectively.
Then, we can express $\hat{Z}^\omega_c(\tau)$ as follows:
\begin{align*}
&\hat{Z}^\omega_c(\tau) = \hat{Y}^\omega_c(\tau) + \hat{X}_{10,c}^\omega(\tau) + \hat{X}_{01,c}^\omega(\tau) = R_1I^1_c(\tau) + R_1I^2_c(\tau) \\
&+ R_2I^3_c(\tau) + R_2I^4_c(\tau) + \hat{X}_{10,c}^\omega(\tau) + \hat{X}_{01,c}^\omega(\tau)\\
&= R_1I^1_c(\tau) + R_1I^2_c(\tau) + R_2I^3_c(\tau) + R_2I^4_c(\tau) + R_1I^1_c(\tau)\\
&+ R_2I^3_c(\tau) = 2R_1I^1_c(\tau) + R_1I^2_c(\tau) + 2R_2I^3_c(\tau) + R_2I^4_c(\tau)
\end{align*}
Note that under any scheduling policy, only one of the decision variables can be $1$ and the rest are
$0$. Thus, the preference order for decisions to maximize $\hat{Z}^\omega_c(\tau)$ is
evident. Specifically, it would be $I_c^1(\tau) \succ I_c^2(\tau)
\succ I_c^3(\tau) \succ I_c^4(\tau)$ when $R_1 \geq 2R_2$ and
$I_c^1(\tau) \succ I_c^3(\tau) \succ I_c^2(\tau) \succ I_c^4(\tau)$
when $R_2 \leq R_1 < 2R_2$. Thus, in each cell $c$, $\hat{Z}^\omega_c (\tau)$ is
maximized by the policy $\omega$ that chooses the scheduling decisions in this preference order,
choosing a less preferred decision only when none of the more
preferred decisions are possible in that cell.

Define ${Z}_c(\tau) = \max_{\omega \in \Psi} \hat{Z}^\omega_c(\tau)$. Then using
(\ref{eq:two}) and (\ref{eq:three}), we have
\begin{eqnarray*}
\lambda \leq \lim_{T\rightarrow\infty}
\frac{1}{2TN}\sum_{\tau=0}^{T-1}\sum_{c=1}^{C}{Z}_c(\tau)
\end{eqnarray*}
As ${Z}_c(\tau)$ can take only a finite number of values (namely
$R_1,R_2, 2R_1, 2R_2$ and $0$) and is a function of the current
state of the ergodic user location processes, the time average of
${Z}_c(\tau)$ is exactly equal to its expectation with respect
to the steady-state user location distribution. Thus, the bound above
can be computed by calculating the expectation of
${Z}_c(\tau)$ using the steady-state probabilities associated
with the indicator variables and summing over all cells. 
When $R_1 \geq 2R_2$, this is given by:
\begin{align*}
&\lim_{T\rightarrow\infty} \frac{1}{2TN}\sum_{\tau=0}^{T-1}\sum_{c=1}^{C}{Z}_c(\tau) \\
&\qquad = \frac{1}{2N} \sum_{c=1}^C\expect{Z_c(\tau)}\\
&\qquad = \frac{2R_1q + R_1 (p-q) + 2R_2 q' + R_2(p'-q')}{2\theta}
\end{align*}
and when $R_2 \leq R_1 < 2R_2$, this is given by:
\begin{align*}
&\lim_{T\rightarrow\infty} \frac{1}{2TN}\sum_{\tau=0}^{T-1}\sum_{c=1}^{C}{Z}_c(\tau) \\
&\qquad = \frac{1}{2N} \sum_{c=1}^C\expect{Z_c(\tau)}\\
&\qquad = \frac{2R_1q + 2R_2 q'' +  R_1 p'' + R_2(p'-q')}{2\theta}
\end{align*}

This establishes the necessary condition for the network capacity.  
\end{proof}

Note that the above preference order clearly spells out the
structure of the throughput optimal strategy. Specifically,
depending on the values of $R_1$ and $R_2$, this order can be used
to decide between same cell relay and adjacent cell direct
transmission. We use this insight to design a throughput-optimal 
$2$-hop relay algorithm in the next section.
Also note the factor of $2$ with the decision variables
corresponding to direct source-destination transmission.
Intuitively, each such transmission opportunity is better than a
similar opportunity between source-relay or relay-destination by a
factor of $2$ since the indirect transmissions need twice as many
opportunities to deliver a given number of packets to the
destination as compared to direct transmissions.

\subsection{Proof of Sufficiency} 
\label{section:2hop_relay}

Now we present an algorithm that makes stationary, randomized scheduling decisions independent of the queue backlogs and show that
it gives bounded delay for any rate $\lambda < \mu$, i.e., 
there exists a $\rho$ such that $0 \leq \rho < 1$ and $\lambda = \rho\mu$. 
We only consider the case when $R_1 \geq 2R_2$. The other case is
similar and is not discussed.

{\bf{2-Hop Relay Algorithm:}} Every timeslot, for all
cells, do the following:
\begin{enumerate}
\item If there exists a source-destination pair in the cell,
randomly choose such a pair (uniformly over all such pairs in the
cell). If the source has new packets for the destination, transmit
at rate $R_1$. Else remain idle.
\item If there is no source-destination pair in the cell but there
are at least $2$ users in the cell, randomly designate one user as
the sender and another as the receiver. 
Then, with probability $\frac{1-\delta}{2}$
(where $0 < \delta < 1$ and  $\delta = \delta(\rho)$ to be determined later),
perform the first action
below. Else, perform the second.

\begin{enumerate}
\item \emph{Send new Relay packets in same cell}: If the transmitter
has new packets for its destination, transmit at rate $R_1$. Else
remain idle.
\item \emph{Send Relay packets to their Destination in same
cell}: If the transmitter has packets for the receiver, transmit at
rate $R_1$. Else remain idle.
\end{enumerate}
\item If there is only $1$ user in the cell and its destination is
present in one of the adjacent cells, transmit at rate $R_2$ if new
packets present. Else remain idle.
\item If there is only $1$ user in the cell and its destination is
not present in one of the adjacent cells but there is at least one
user in an adjacent cell, randomly designate one such user as the
receiver and the only user in the cell as the transmitter. 
Then, with probability $\frac{1-\delta}{2}$
(where $0 < \delta < 1$ and  $\delta = \delta(\rho)$ to be determined later),
perform the first action
below. Else, perform the second.
\begin{enumerate}
\item \emph{Send new Relay packets in adjacent cell}: If the transmitter has
new packets for its destination, transmit at rate $R_2$. Else remain
idle.
\item \emph{Send Relay packets to their Destination in adjacent
cell}: If the transmitter has packets for the receiver, transmit at
rate $R_2$. Else remain idle.
\end{enumerate}
\end{enumerate}
This algorithm is motivated by the proof of necessity of Theorem
\ref{thm:capacity} since it follows the same preference order in
making scheduling decisions. Note that this algorithm restricts the
path lengths of all packets to \emph{at most} $2$ hops because any
packet that has been transmitted to a relay node is restricted from
being transmitted to any other node except its destination.

To analyze the performance of this algorithm, we make use of a
Lyapunov drift analysis \cite{neely-thesis}. Consider a
network of $N$ queues operating in slotted time, and let
$\vec{U}(t) = (U_1(t),U_2(t),\ldots,U_N(t))$ represent the
vector of backlogged packets in each of the queues at timeslot $t$.
Let $L(\vec{U}(t))$ be a
non-negative function of the unfinished work $\vec{U}(t)$,
called a Lyapunov function. Define the conditional Lyapunov drift
$\Delta(t, d)$ at time $t > d$ (where $d \geq 0$ in a fixed integer) as follows:
\begin{align*}
\Delta(t, d) \defequiv \expect{L(\vec{U}(t+1)) - L(\vec{U}(t)) | \vec{U}(t-d)}
\end{align*}
Then we have the following lemma.
\begin{lem} \emph{Lyapunov Drift Lemma}: If there exists a positive
integer $d$ such that for all timeslots $t > d$ and for all
$\vec{U}(t)$, the conditional Lyapunov drift 
$\Delta(t, d)$
satisfies:
\begin{align}
\Delta(t, d) \leq B - \epsilon \sum_{i=1}^N {U_i(t-d)} 
\label{eq:ldlemma}
\end{align}
for some positive constants $B$ and $\epsilon$, and if
$\expect{L(\vec{U}(d))} < \infty$, then the network is stable,
and we have the following bound on the time average total queue
backlog:
\begin{align}
\limsup_{t\rightarrow\infty}
\frac{1}{t}\sum_{\tau=0}^{t-1}\sum_{i=1}^N \expect{U_i(\tau)} \leq
\frac{B}{\epsilon}
\label{eq:delay_bound0}
\end{align}
\end{lem}
\begin{proof} This can be shown using a telescoping sum technique (similar to 
related proof in \cite{neely-thesis}) and is omitted for brevity.
\end{proof}

We now make use of this lemma to analyze the performance of the $2$-Hop
Relay Algorithm.

\begin{thm} For the cell partitioned network (with $N$ nodes
and $C$ cells) as described in Sec. \ref{section:model}, with capacity $\mu = \frac{R_1p + R_2p' + R_1q + R_2q'}{2 \theta}$ and input
rate $\lambda$ for each user such that 
$\lambda = \rho\mu$ for some $0 \leq \rho < 1$, 
and user mobility model as described in Sec. \ref{section:mob_model},
 the average packet delay $\overline{D}$ under the $2$-Hop Relay Algorithm with $\delta = \frac{1 - \rho}{4}$ satisfies:
\begin{align}
\overline{D} \leq \frac{{BN(2d+1)}}{\lambda\mu\kappa(1-\rho)}
\label{eq:delay_bound1}
\end{align}
where ${B}$ is a constant given by (\ref{eq:B}), $\kappa$ is a positive constant given by 
$\kappa = \frac{R_1p + R_2p' - R_1q - R_2q'}{R_1p + R_2p' + R_1q + R_2q'}$, 
 and $d$ is a finite integer that is related to the
mixing time of the joint user mobility process and is given by $d =
\Big\lceil{\frac{\log( \frac{8 N^2\alpha}{1-\rho})}{\log(1/\gamma)}}\Big\rceil$.
\label{thm:delay_bound} 
\end{thm}
\begin{proof} 
Let $U^{(c)}_i(t)$ represent the total backlog of type
$c$ (i.e., number of packets destined for node $c$) that are queued
up in node $i$ at time $t$. The queueing dynamics of $U^{(c)}_i(t)$
satisfies the following for all $c \neq i$:
\begin{align}
U^{(c)}_i(t+1) \leq \max &\Big[ U^{(c)}_i(t) - \sum_{b} \mu_{ib}^{(c)}(t), 0 \Big]  + \sum_{a} \mu_{ai}^{(c)}(t) \nonumber \\ 
&+ A_{i}^{(c)}(t) 
\label{eq:four}
\end{align}
where $A_{i}^{(c)}(t)=$number of new type $c$ arrivals to source node $i$ at the beginning of timeslot $t$ and
$\mu_{ab}^{(c)}(t)=$rate offered to type $c$ packets in timeslot $t$ with node $a$
as transmitter and node $b$ as receiver. The above is an inequality because the actual number of packets
transmitted from the other nodes to node $i$ (for relaying) could be
less than the incoming transmission rate $\sum_{a} \mu_{ai}^{(c)}(t)$
when these nodes do not have enough packets. Now define a Lyapunov function $L(\vec{U}(t)) = \sum_{i=1}^N \sum_{c\neq i} (U^{(c)}_i(t))^2$. 
Using (\ref{eq:four}), 
the conditional Lyapunov drift $\Delta(t, d)$ can be expressed as follows:
\begin{align}
&\Delta(t, d) \leq BN -2\sum_{i=1}^N \sum_{c\neq i} \mathbb{E} \Bigg\{U^{(c)}_i(t) \times\nonumber \\
& \Big(\sum_b \mu_{ib}^{(c)}(t) - \sum_a \mu_{ai}^{(c)}(t) - A_{i}^{(c)}(t)\Big) |\vec{U}(t-d)\Bigg\}
 \label{eq:five}
\end{align}
Here, $B$ is given by:
\begin{align}
B = (A_{max} + \mu^{in}_{max})^2 + (\mu_{max}^{out})^2
 \label{eq:B}
\end{align}
where $\mu^{in}_{max} =$ maximum transmission rate into any node $= R_1 + JR_2$, where $J$ is the maximum number of adjacent cells of any cell 
(Sec. \ref{section:cell}) and $\mu^{out}_{max} =$ maximum transmission rate out of any node $= R_1$.

We now use the following sample path relations to express
(\ref{eq:five}) in terms of the queue backlog values at time $(t-d)$.
Specifically, we have the following for all $t > d$ where $d$ is a
positive integer (to be determined later) for all $i \neq c$.
\begin{align*}
& \sum_{c\neq i} U^{(c)}_i(t-d) + d(A_{max} + \mu^{in}_{max}) \geq \sum_{c\neq i} U^{(c)}_i(t) \\
& \sum_{c\neq i} U^{(c)}_i(t-d) - d\mu^{out}_{max} \leq \sum_{c\neq i} U^{(c)}_i(t) 
\end{align*}
These follow by noting that the queue backlog at time $t$ cannot be
smaller than the queue backlog at time $(t-d)$ minus the maximum
possible departures in duration $(t-d, d)$. Similarly, it cannot be
larger than the queue backlog at time $(t-d)$ plus the maximum
possible arrivals in duration $(t-d, d)$.
Using these, we can express (\ref{eq:five}) in terms of the ``delayed'' queue backlogs $U^{(c)}_i(t-d)$ as follows:
\begin{align}
&\Delta(t, d) \leq BN(2d+1) - 2\sum_{i=1}^N \sum_{c \neq i} U^{(c)}_i(t-d) \times \nonumber\\
 &\expect{\sum_b \mu_{ib}^{(c)}(t) - \sum_a \mu_{ai}^{(c)}(t) - A_{i}^{(c)}(t)|\vec{U}(t-d)}
 \label{eq:delayed}
\end{align}
Let $\mathcal{T}(t-d) = (\vec{\chi}(t-d), \vec{U}(t-d))$ represent the composite system state at time
$(t-d)$ given by the user locations and queue backlogs. 
Since the $2$-Hop Relay Algorithm makes control decisions only as a function of the
current user locations, the resulting service rates
are functionals of the Markovian mobility processes. 
By the Markovian property of the $\vec{\chi}(t-d)$ process, any
functionals of this also converge exponentially fast to their steady-state values. 
Thus, using Lemma \ref{lem:lem1}, when $\alpha \gamma^d \leq 1/N^2$, we have the following bounds:
\begin{align}
&\expect{\sum_b \mu_{ib}^{(c)}(t)|\vec{U}(t-d)} \geq \Big( \sum_b \overline{\mu}_{ib}^{(c)} \Big) (1 -  2N\alpha \gamma^d ) 
\label{eq:stat6-1} \\ 
&\expect{\sum_a \mu_{ai}^{(c)}(t) |\vec{U}(t-d)} \leq \Big(\sum_a \overline{\mu}_{ai}^{(c)} \Big) (1 + 2N\alpha \gamma^d ) 
\label{eq:stat6-2}
\end{align}
where $\overline{\mu}_{ib}^{(c)}, \overline{\mu}_{ai}^{(c)}$ are the
steady-state service rates achieved by the $2$-Hop Relay Algorithm.
We now compute these values and use the inequalities (\ref{eq:stat6-1}), (\ref{eq:stat6-2}) to obtain a bound on (\ref{eq:delayed}). 
We have the following $2$ cases:

\emph{$1)$ Node $i$ Generates Type $c$ Packets:} In this case,
$\expect{A_{i}^{(c)}(t)} = \lambda$ and 
$\sum_a \mu_{ai}^{(c)}(t) = 0$ (since under the $2$-Hop Relay Algorithm, a source node would never get
back a packet that it generates). To calculate $\sum_b
\overline{\mu}_{ib}^{(c)}$, we note that
the outgoing service rate for packets generated by the source is
equal to the sum of the rate at which the source is scheduled to
transmit directly to its destination and the rate at which it is
scheduled to transmit type $c$ packets to any of the relay nodes. Let these rates be
$r_1$ and $r_2$ respectively. Also let the transmission rate at which it is scheduled to 
transmit relay packets to their destinations be $r_3$. 
Since the $2$-Hop Relay Algorithm only schedules
transmissions of these types,
the total rate of transmissions over the
network is given by $N(r_1 + r_2 + r_3)$. Using the probability of
choosing source-relay and relay-destination transmissions, we have:
$r_2 = \frac{1-\delta}{1 + \delta}r_3$.
In the $2$-Hop Relay Algorithm, a direct source-to-destination
transmission is scheduled whenever there is a source-destination
pair in the same cell or there is only $1$ node in a cell and its
destination is in an adjacent cell (and independent of the actual queue backlog values). Thus, 
using the definitions of $q$ and $q'$ from the statement of
Theorem \ref{thm:capacity}, we have:
$Nr_1 = C(R_1q + R_2q')$.
Similarly, the sum total transmissions in the network can be
expressed in terms of the quantities $p$ and $p'$ as follows:
$N(r_1 + r_2 + r_3) = C(R_1p + R_2p')$.
Using these to solve for $r_1, r_2, r_3$ and simplifying, we have 
\begin{align}
&r_1 = \mu (1 - \kappa), \;\; r_2 = \mu \kappa (1 - \delta), \;\; r_3 = \mu \kappa (1 + \delta)
\label{eq:r1}
\end{align}
where $\kappa \defequiv \frac{R_1p + R_2p' - R_1q - R_2q'}{R_1p + R_2p' + R_1q + R_2q'}$. Note that $0 < \kappa < 1$ (since $p > q$ and $p' > q'$).
Therefore, we have:
\begin{align*}
\sum_b \overline{\mu}_{ib}^{(c)} = r_1 + r_2 = \mu - \mu \kappa \delta 
\end{align*}
Let $\delta = \frac{1-\rho}{4}$ and  $\alpha \gamma^d = \frac{\delta}{2N^2} = \frac{1-\rho}{8N^2}$. 
Note that this choice of $\delta$ represents a valid probability since $0 \leq \rho < 1$.
Then, using (\ref{eq:stat6-1}), the last term of (\ref{eq:delayed}) under this case can be expressed as:
\begin{align*}
&\expect{\sum_b \mu_{ib}^{(c)}(t) - \sum_a \mu_{ai}^{(c)}(t) - A_i^{(c)}(t) |\vec{U}(t-d)} \geq \\
&\Big( \sum_b \overline{\mu}_{ib}^{(c)} \Big) (1 -  2N\alpha \gamma^d) - \lambda =  (\mu - \mu \kappa \delta)(1 - \frac{\delta}{N}) - \rho\mu \\
& \geq \mu \Big[ (1-\delta)^2 - \rho \Big] \geq \mu(1 - 2 \delta - \rho) = \frac{\mu(1 - \rho)}{2}
\end{align*}
where we used the fact that $(1 - \kappa \delta)(1 - \frac{\delta}{N}) \geq (1 - \delta)^2$.

\emph{$2)$ Node $i$ Relays Type $c$ Packets:} Note that $N > 2$ for this case to happen.
From our traffic model, we know that in this case
$A_{i}^{(c)}(t) = 0$ for all $t$. Further, under the
$2$-Hop Relay Algorithm, $\mu_{ai}^{(c)}(t)
> 0$ only if node $a$ is the source for type $c$ packets.
Also $\mu_{ib}^{(c)}(t) > 0$ only if $b = c$. 
To compute $\sum_b \overline{\mu}_{ib}^{(c)}$ and $\sum_a \overline{\mu}_{ai}^{(c)}$ for this case, note
that the $2$-Hop Relay Algorithm schedules relay transmissions such that all $(N-2)$ relay packet types are
equally likely. Thus we have:
\begin{align*}
\sum_b \overline{\mu}_{ib}^{(c)} = \frac{r_3}{N-2}, \;\; \sum_a \overline{\mu}_{ai}^{(c)} = \frac{r_2}{N-2}
\end{align*}
Let $\delta = \frac{1-\rho}{4}$ and  $\alpha \gamma^d = \frac{\delta}{2N^2} = \frac{1-\rho}{8N^2}$. 
Then, using (\ref{eq:stat6-1}), (\ref{eq:stat6-2}), the last term of (\ref{eq:delayed}) under this case can be expressed as:
\begin{align*}
&\expect{\sum_b \mu_{ib}^{(c)}(t) - \sum_a \mu_{ai}^{(c)}(t) - A_i^{(c)}(t) |\vec{U}(t-d)} \\
&\geq \Big( \sum_b \overline{\mu}_{ib}^{(c)} \Big) (1 -  2N\alpha \gamma^d ) - \Big( \sum_a \overline{\mu}_{ai}^{(c)} \Big) (1 +  2N\alpha \gamma^d ) \\
&= \Big( \sum_b \overline{\mu}_{ib}^{(c)} -\sum_a \overline{\mu}_{ai}^{(c)} \Big) - \Big( \sum_b \overline{\mu}_{ib}^{(c)} + 
\sum_a \overline{\mu}_{ai}^{(c)} \Big) \frac{\delta}{N} \\
&= \frac{(r_3 - r_2) - \frac{(r_3+r_2)\delta}{N}}{N-2}
= \frac{2\mu \kappa\delta}{N-2} \Big(1 - \frac{1}{N}\Big) \geq \frac{ \mu \kappa (1-\rho)}{2N}
\end{align*}
where we used (\ref{eq:r1}). 
Combining these two cases,  with $\delta =  \frac{1 - \rho}{4}$ and $\alpha \gamma^d = \frac{1-\rho}{8N^2}$:
\begin{align*}
\expect{\sum_b \mu_{ib}^{(c)}(t)-\!\sum_a \mu_{ai}^{(c)}(t) -\!A_{i}^{(c)}(t)|\vec{U}(t-d)}\!\geq \frac{\mu \kappa (1-\rho)}{2N}
\end{align*}
Using this in (\ref{eq:delayed}), we get:
\begin{align*}
\Delta(t, d) \leq BN(2d+1)  
- \frac{\mu \kappa (1-\rho)}{N} \sum_{i=1}^N\sum_{c \neq i} {U^{(c)}_i(t-d)}
\end{align*}
This is in a form that fits (\ref{eq:ldlemma}). Using the Lyapunov
Drift Lemma, we get
\begin{align}
\limsup_{t\rightarrow\infty} \frac{1}{t}\sum_{\tau=0}^{t-1}\sum_{i
\neq c}\expect{U_i^{(c)}(\tau) } \leq \frac{BN^2(2d+1)}{\mu \kappa (1-\rho)}
\label{eq:delay_bound2}
\end{align}
The total input rate into the network is $N\lambda$. Thus,  using
Little's Theorem, the average delay per packet is bounded by
$\frac{{BN(2d+1)}}{\lambda\mu\kappa(1-\rho)}$. 
\end{proof}

\subsection{Discussion and Simulation Example}

The proof of the capacity for the
cell-partitioned network can be used to consider more general
scheduling restrictions. From (\ref{eq:three}), it amounts to:
\begin{align*}
\lambda \leq \frac{1}{2N}\expect{\max_{\omega\in \Psi}
\sum_{c=1}^{C}\Big(Y_c^\omega(t) + X_{10,c}^\omega(t) + X_{01,c}^\omega(t)\Big)}
\end{align*}
If the bound on the right hand side can be achieved by any policy
(potentially randomized) that takes decisions only as a function of
the current network state, then we can design a deterministic policy
that is throughput optimal by scheduling to maximize
$\sum_{c=1}^C\Big(Y_c^\omega(t) + X_{10,c}^\omega(t) + X_{01,c}^\omega(t)\Big)$
subject to the network restrictions. For the specific
cell-partitioned model considered here, this maximization is
achieved by following the preference order of the decision variables
in each cell separately as described earlier. This enables us to exactly compute the
capacity of the network. It is possible to do the same for
extensions to this model involving other constraints. For example,
under the constraint that a user cannot simultaneously
transmit and receive, the above maximization becomes a
maximum-weight match problem. 
Similarly, one could allow more than one transmitter per cell, in which case we would need to define
more indicator decision variables for all possible control options.


\begin{figure}
\centering
\includegraphics[width=9cm, height=4.5cm, angle=0]{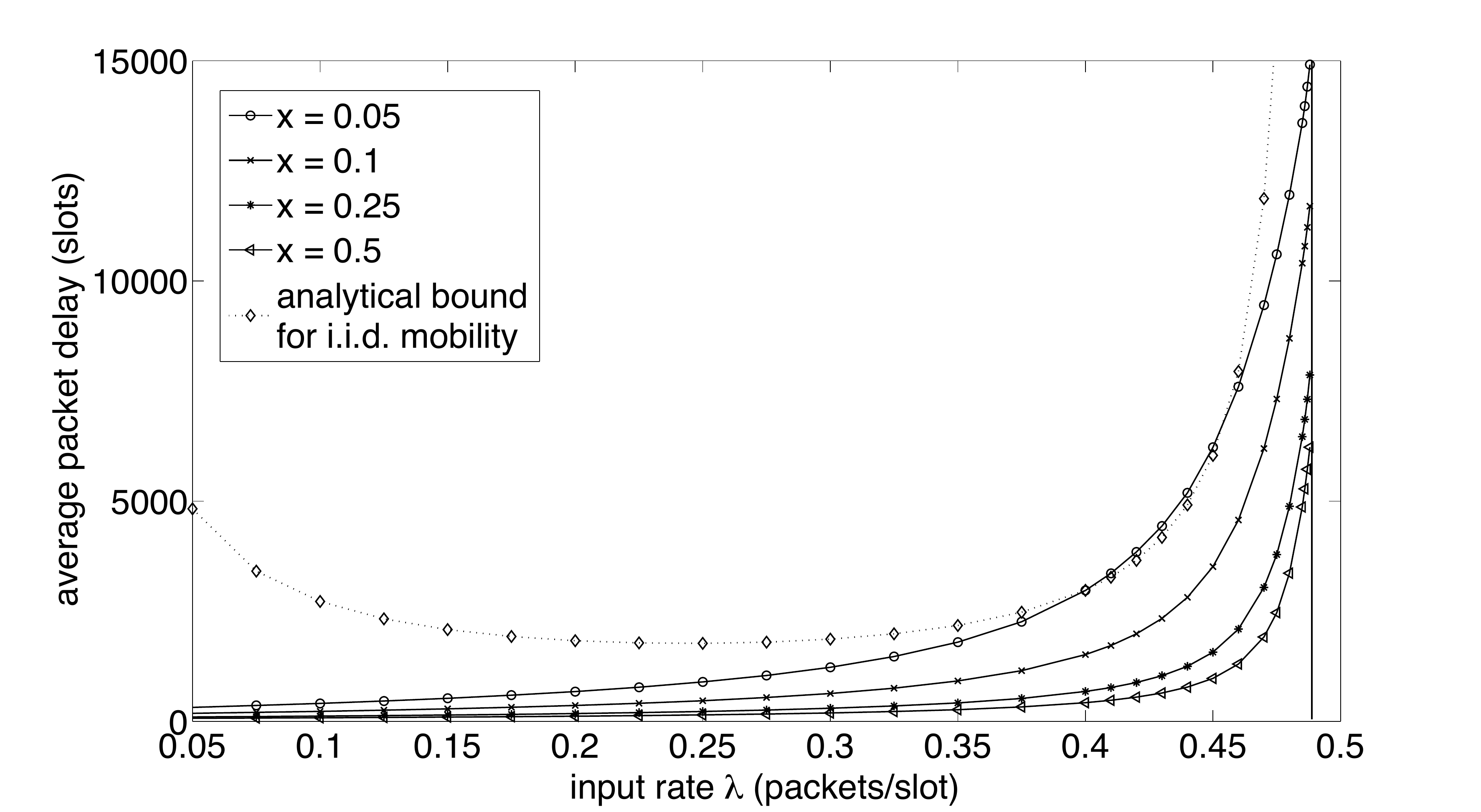}
\caption{Average packet delay under the $2$-Hop Relay Algorithm in a
network of $16$ cells with $20$ nodes for different mixing times of the mobility process.}
\label{fig:2}
\end{figure}

We next consider an example network consisting of $20$ nodes and $16$ cells as shown in Fig. \ref{fig:1}. 
The nodes move from one cell to another independently according to a Markovian random walk. 
Specifically, at the end of every slot, a node stays in its current cell with probability $(1-x)$, else it 
decides to move randomly one step in either the North, West, South, or East directions with probability $x$. 
If there is no feasible adjacent cell, 
then the user remains in the current cell. 
It can be shown that the resulting steady-state location distribution is uniform over all cells for all $0 \leq x < 1$. 
Thus, $\pi_c = \frac{1}{16}$ for all cells $c$.
Next we assume that $R_1 = 2$ and $R_2 = 1$ packets/slot. Then using Theorem \ref{thm:capacity}, the capacity for this network 
is given by $\mu = \frac{R_1q + R_1p + R_2q' + R_2p'}{2 \theta}$ and can be calculated exactly.
Specifically, we get $p = 0.358, q = 0.038, p' = 0.357, q' = 0.073$ and the network capacity is given by $\mu = 0.489$ packets/slot.

We next simulate the $2$-Hop Relay Algorithm on this network. 
New packets arrive at each source node according
to independent Bernoulli processes, so that a single
packet arrives i.i.d. with probability $\lambda$ every slot.
In Fig. \ref{fig:2}, we plot the average packet delay vs. $\lambda$ for different
values of $x$. We also plot the analytical bound (\ref{eq:delay_bound1}) of Theorem \ref{thm:delay_bound}
for the i.i.d. mobility case (for which $d = 0$).
It can be seen that the average delay goes to infinity as $\lambda$ is
pushed closer to the capacity $\mu = 0.489$ packets/slot (shown by the vertical line in Fig. \ref{fig:2}). 
While the network capacity is the same for
all values of $x$ (since $x$ does not affect the steady-state location distribution),  the average delay increases 
as $x$ becomes smaller. This is because a smaller $x$ implies a larger value for the parameter $d$ leading to
larger delay as suggested by the delay bound (\ref{eq:delay_bound1}) in Theorem \ref{thm:delay_bound}. 
Thus, the $2$-Hop Relay Algorithm is able to support all input rates within the network capacity with 
finite average delay. However, its delay performance is not necessarily the best. For example,
when the input rate is small (say $\lambda = 0.1$ packets/slot), the average delay is more than $100$ slots. 
Note that the $2$-Hop Relay Algorithm makes scheduling
decisions purely based on the current user locations and restricts all
packets to at most $2$ hops. 
It does not attempt to optimize the delay in the network. 
The delay performance may be improved using alternative scheduling strategies
that do not restrict packets to at most $2$ hops.
For example, backlog aware scheduling and routing (e.g., \cite{neely-thesis}) or schemes that exploit the mobility pattern of
the users (e.g., \cite{grossglauser-vetterli}) may offer better delay performance.

\section{Minimum Energy Function}
\label{section:energy}
We now investigate the minimum energy function of the
cell-partitioned network under consideration. Recall that in our network model,
each user either uses zero power or full power. Further, 
$R_1 (R_2)$ packets can
be transmitted from the sender to the receiver in the same (adjacent) cell if the sender
uses full power. 

The minimum energy
function $\Phi(\lambda)$ is defined as the minimum time average
energy required to stabilize an input rate $\lambda$ per user,
considering {all} possible scheduling and routing algorithms that
conform to the given network structure. We exactly compute this
function for our network model. Specifically, we assume that all
users receive packets at the same rate (i.e., $\lambda_i = \lambda$
for all $i$). Also, we consider the case when $R_1 \geq 2R_2$ ($\Phi(\lambda)$ for the case when $R_1 < 2R_2$ has a different
expression, but the proof is similar).

\begin{thm} The minimum energy function $\Phi(\lambda)$ per user for
the cell-partitioned network as described in Sec. \ref{section:model} with $R_1 \geq 2R_2$ 
is a piecewise linear curve given by the following:
\begin{displaymath}
\Phi(\lambda) = \left\{ \begin{array}{ll}
\frac{\lambda}{R_1} & \textrm{if $C_1$}\\
\frac{q}{\theta} + \frac{2}{R_1}\Big(\lambda - \frac{2R_1q}{2\theta}\Big) & \textrm{if $C_2$}\\
\frac{p}{\theta} + \frac{1}{R_2}\Big(\lambda - \frac{R_1(p+q)}{2\theta}\Big) & \textrm{if $C_3$}\\
\frac{p+q'}{\theta} + \frac{2}{R_2}\Big(\lambda - \frac{R_1(p+q)+2R_2q'}{2\theta}\Big) & \textrm{if $C_4$}
\end{array} \right.
\end{displaymath}
\label{thm:energy}
\end{thm}
where 
$C_1 \equiv 0 \leq \lambda < \frac{R_1q}{\theta}, C_2 \equiv \frac{R_1q}{\theta} \leq \lambda < \frac{R_1(p+q)}{2\theta},
C_3 \equiv \frac{R_1(p+q)}{2\theta} \leq \lambda < \frac{R_1(p+q)+2R_2q'}{2\theta},
C_4 \equiv \frac{R_1(p+q)+2R_2q'}{2\theta} \leq \lambda < \mu$.
Thus, the network can stably support users simultaneously
communicating at any rate $\lambda < \mu$ with an energy cost that
can be pushed arbitrarily close to $\Phi(\lambda)$ (at the cost of
increased delay). We prove the theorem in two parts. First, we
establish the necessary condition by deriving a lower bound on the
energy cost of \emph{any} stabilizing algorithm. Then, we establish
sufficiency by presenting a specific scheduling policy and showing
that the average delay is bounded under that policy.

\subsection{Proof of Necessity} 
\label{section:energy_proof}

\begin{proof}
Consider any scheduling strategy that stabilizes the system. Let $X_{ab}(T)$ denote the
number of packets delivered by the strategy from sources to
destinations in time interval $(0, T)$ that involves exactly $a$
same cell and $b$ adjacent cell transmissions. For simplicity,
assume that the strategy is ergodic and yields well defined time
average energy expenditure per user $\overline{e}$ and well defined time average values
for $x_{ab}$ where:
\begin{align}
x_{ab} \ \defequiv \ \lim_{T\rightarrow\infty} \frac{X_{ab}(T)}{T}
\end{align}
The average energy cost per user $\overline{e}$ of this policy satisfies:
\begin{align}
\overline{e} \geq \sum_{a,b}\Big(\frac{a}{R_1} + \frac{b}{R_2}\Big)\frac{x_{ab}}{N}
 \label{eq:nine}
\end{align}
This follows by noting that enough packets may not be available
during a transmission.

Note that $x_{00} = 0$, and so the only possible non-zero $x_{ab}$
variables are for $(a, b)$ pairs that are integers, non-negative,
and such that $(a, b) \neq (0, 0)$. Let $x = (x_{ab})$ represent the
collection of $x_{ab}$ variables, and note that these variables must
satisfy the constraint $x \in \Omega_0 \cap \Omega_1 \cap \Omega_2
\cap \Omega_3$, where the four constraint sets are defined below:
\begin{align*}
& \Omega_0 \defequiv \Bigg\{x \Bigg| \sum_{(a,b) \neq (0,0)}x_{ab} = N\lambda\Bigg\} 
\qquad \Omega_1 \defequiv \Bigg\{x \Bigg| \frac{x_{10}}{R_1} \leq c_1 \Bigg\} \\
& \Omega_2 \defequiv \Bigg\{x \Bigg| \frac{1}{R_1}\sum_{a}ax_{a0} \leq c_1 + c_2 \Bigg\} \\
& \Omega_3 \defequiv \Bigg\{x \Bigg| \frac{1}{R_1}\sum_{a}ax_{a0} + \frac{x_{01}}{R_2} \leq c_1 + c_2 + c_3\Bigg\}
\end{align*}
where $c_1$ is the maximum rate of source-destination transmission
opportunities in the same cell, $c_1 + c_2$ is the maximum
rate of all possible same cell transmission opportunities 
and $c_1 + c_2 + c_3$ is the maximum rate of all same cell or
source-destination adjacent cell transmission opportunities. Here,
these quantities are summed over all cells.
Using the definitions of $p, q$ and $q'$ from the statement of Theorem \ref{thm:capacity}, we know that 
$c_1 = Cq, c_1 + c_2 = Cp, c_1 + c_2 + c_3 = C(p +
q')$. For example, $(c_1 + c_2 + c_3)$ can be written as
$\frac{1}{T} \sum^T_{t=0} \Big( X_1(t) + X_2(t) + X_3(t) \Big)$ where $X_1(t)$
is the maximum number of direct same cell opportunities, $X_2(t)$ is
the maximum number of indirect same cell opportunities \emph{given}
all direct opportunities are used and $X_3(t)$ is the maximum number
of direct adjacent cell opportunities \emph{given} all same cell
opportunities are used. Since only one of these three opportunities
can used is a given cell in a timeslot, the maximum total sum is
fixed and hence $c_1 + c_2 + c_3 = C(p + q')$.

Define $f(x) \defequiv \sum_{a,b}\Big(\frac{a}{R_1} + \frac{b}{R_2}\Big)\frac{x_{ab}}{N}$, which is simply the right hand side of
(\ref{eq:nine}). Because $\overline{e}\geq f(x)$, and because $x \in \Omega_0
\cap \Omega_1 \cap \Omega_2 \cap \Omega_3$, we have:
\begin{align}
\overline{e} \geq \inf_{x\in \Omega_0 \cap \Omega_1 \cap \Omega_2 \cap
\Omega_3} f(x)
 \label{eq:10}
\end{align}
Furthermore, for any function $g(x)$ such that $g(x) \leq f(x)$ for
all $x$, and for any set $\tilde{\Omega}$ that contains the set
$\Omega_0 \cap \Omega_1 \cap \Omega_2 \cap \Omega_3$, we have:
\begin{align}
\overline{e} \geq \inf_{x\in \tilde{\Omega}} g(x)
 \label{eq:11}
\end{align}
This follows because the function to be minimized is smaller, and
the infimum is taken over a less restrictive set.
We now define four new constraint sets $\tilde{\Omega}_0,
\tilde{\Omega}_1, \tilde{\Omega}_2, \tilde{\Omega}_3$ as follows:
\begin{align*}
& \tilde{\Omega}_0 \defequiv \Omega_0 \; \tilde{\Omega}_1 \defequiv \Omega_1 \; 
\tilde{\Omega}_2 \defequiv \Bigg\{x \Bigg| \frac{x_{10}}{R_1}+ \frac{2}{R_1}\sum_{a \geq 2}x_{a0} \leq c_1 + c_2 \Bigg\} \\
& \tilde{\Omega}_3 \defequiv \Bigg\{x \Bigg| \frac{x_{10}}{R_1}+ \frac{2}{R_1}\sum_{a \geq 2}x_{a0} + \frac{x_{01}}{R_2} \leq c_1 + c_2 + c_3 \Bigg\}
\end{align*}
It can be seen that each of $\Omega_0, \Omega_1, \Omega_2, \Omega_3$
is a subset of $\tilde{\Omega}_0, \tilde{\Omega}_1,
\tilde{\Omega}_2, \tilde{\Omega}_3$. Therefore, $\Omega_0 \cap
\Omega_1 \cap \Omega_2 \cap \Omega_3$ is a subset of
$\tilde{\Omega}_0 \cap \tilde{\Omega}_1 \cap \tilde{\Omega}_2 \cap
\tilde{\Omega}_3$. 
Note that since $\frac{2}{R_1} \leq \frac{1}{R_2}$, 
we have the following:
\begin{align}
\frac{1}{R_1} < \frac{2}{R_1} \leq \frac{1}{R_2} < \frac{2}{R_2}
 \label{eq:12}
\end{align}
We now compute four different bounds for $\overline{e}$, each
having the form $\overline{e} \geq \alpha \lambda + \beta$. These bounds define
the four piecewise linear regions of $\Phi(\lambda)$. 

\begin{enumerate}

\item First note that $f(x) \geq \frac{1}{R_1}\sum_{a,b} \frac{x_{ab}}{N}$. This follows from the
definition of $f(x)$.
Therefore taking $g(x) = \frac{1}{R_1}\sum_{a,b} \frac{x_{ab}}{N}$, we have:
\begin{align*}
\overline{e} \geq \inf_{x \in \tilde{\Omega}_0} \frac{1}{R_1}\sum_{a,b} \frac{x_{ab}}{N}
\end{align*}
Because $\tilde{\Omega}_0$ is given by $\sum_{a,b}x_{ab} = N\lambda$, the above infimum is equal to
$\frac{\lambda}{R_1}$. Thus, we have our first linear constraint
for any algorithm that yields a time average energy of $\overline{e}$:
\begin{align}
\overline{e} \geq \frac{\lambda}{R_1}
 \label{eq:13}
\end{align}

\item Next note that $f(x) \geq \frac{x_{10}}{N R_1} + \frac{2}{R_1}\sum_{\substack{a,b \\(a,b)\neq (1,0)}} \frac{x_{ab}}{N}$.
This is because $\frac{a}{R_1} + \frac{b}{R_2} \geq \frac{2}{R_1}$
for any non-negative integer pair $(a, b)$ such that $(a, b) \neq
\{(0, 0), (1, 0)\}$ (using (\ref{eq:12})). 
Therefore, taking this lower bound of $f(x)$ as $g(x)$, we have:
\begin{align*}
\overline{e} \geq \inf_{x \in \tilde{\Omega}_0 \cap \tilde{\Omega}_1} 
\Bigg[\frac{x_{10}}{N R_1} + \frac{2}{R_1}\sum_{\substack{a,b \\(a,b)\neq (1,0)}} \frac{x_{ab}}{N} \Bigg]
\end{align*}
The right hand side is equal to the solution of the following:
\begin{align*}
\textrm{Minimize:} \qquad & \frac{x_{10}}{N R_1} + \frac{2}{R_1}\sum_{\substack{a,b \\(a,b)\neq (1,0)}} \frac{x_{ab}}{N} \\
\textrm{Subject to:} \qquad & \sum_{a,b} x_{ab} = N\lambda \\
& \frac{x_{10}}{R_1} \leq c_1
\end{align*}
The above optimization is equivalent to minimizing
$\frac{x_{10}}{NR_1} + \frac{2}{NR_1}(N\lambda - x_{10})$ subject to
$\frac{x_{10}}{R_1} \leq c_1$. The solution is clearly to choose
$x_{10} = R_1c_1$, and hence we have:
\begin{align}
\overline{e} \geq \frac{2\lambda}{R_1} - \frac{c_1}{N} = \frac{q}{\theta} +
\frac{2}{R_1}\Big(\lambda - \frac{2R_1q}{2\theta}\Big)
 \label{eq:14}
\end{align}

\item  Next we have 
\begin{align*}
f(x) \geq \frac{x_{10}}{N R_1} + \frac{2}{R_1}\sum_{{a \geq 2}} \frac{x_{a0}}{N} + \frac{1}{R_2}\sum_{\substack{a,b\\b \neq 0}} \frac{x_{ab}}{N}
\end{align*}
which follows from the definition of $f(x)$
and because $\frac{1}{R_2} \leq \frac{b}{R_2}$ for all positive $b \geq 1$.
Thus, taking this lower bound of $f(x)$ as $g(x)$, we have:
\begin{align*}
\overline{e} \geq \inf_{x \in \tilde{\Omega}_0 \cap \tilde{\Omega}_1 \cap
\tilde{\Omega}_2} \Bigg[ \frac{x_{10}}{NR_1} + \frac{2}{R_1}\sum_{{a
\geq 2}} \frac{x_{a0}}{N} + \frac{1}{R_2}\sum_{\substack{a,b
\\b \neq 0}} \frac{x_{ab}}{N} \Bigg]
\end{align*}
This is equivalent to the following minimization:
\begin{align*}
\textrm{Minimize:} &\;  \frac{x_{10}}{NR_1} + \frac{2}{NR_1}\sum_{{a \geq 2}} x_{a0} \\
&\qquad + \frac{1}{NR_2}\Big(N\lambda - x_{10} - \sum_{{a \geq 2}} x_{a0}\Big) \\
\textrm{Subject to:}& \qquad \qquad \frac{x_{10}}{R_1} \leq c_1 \\ 
&\; \frac{x_{10}}{R_1} + \frac{2}{R_1}\sum_{{a \geq 2}}x_{a0} \leq c_1 + c_2
\end{align*}
where we have aggregated the constraint $\sum_{a,b} x_{ab} =
N\lambda$ into the objective. The coefficients multiplying $x_{10}$ and $\sum_{a\geq 2}x_{a0}$ are both negative,
so that the
above optimization is solved when $x_{10} + 2\sum_{a\geq 2}x_{a0} = R_1(c_1 + c_2)$. Similarly, it can be
shown that above optimization is solved when $x_{10} = R_1c_1$.
This yields:
\begin{align}
\overline{e} &\geq \frac{\lambda}{R_2} + \frac{(c_1 + c_2)}{N} - \frac{R_1}{NR_2}\Big(c_1 + \frac{c_2}{2}\Big) \nonumber \\
&= \frac{p}{\theta} + \frac{1}{R_2}\Big(\lambda - \frac{R_1(p+q)}{2\theta}\Big)
 \label{eq:15}
\end{align}

\item Finally, note that 
\begin{align*}
f(x) \geq \frac{x_{10}}{NR_1} + \frac{2}{R_1}\sum_{{a \geq 2}} \frac{x_{a0}}{N} + \frac{x_{01}}{N R_2} + \frac{2}{R_2}\sum_{b \geq 2} \frac{x_{ab}}{N}
\end{align*}
which follows from the definition of $f(x)$
as well as because $\frac{2}{R_2} \leq \frac{b}{R_2}$ for all $b \geq 2$. 
Taking this lower bound of $f(x)$ as $g(x)$, we have:
\begin{align*}
\overline{e} \geq \inf_{x \in \tilde{\Omega}} \Bigg[ \frac{x_{10}}{NR_1} +
\frac{2}{R_1}\sum_{a \geq 2} \frac{x_{a0}}{N} + \frac{x_{01}}{NR_2} +
\frac{2}{R_2}\sum_{b \geq 2} \frac{x_{ab}}{N} \Bigg]
\end{align*}
where  $\tilde{\Omega} = \tilde{\Omega}_0 \cap \tilde{\Omega}_1 \cap
\tilde{\Omega}_2 \cap \tilde{\Omega}_3$. This is equivalent to the
following minimization (using $\sum_{a,b} x_{ab} = N\lambda$):
\begin{align*}
\textrm{Minimize:} \; &\frac{x_{10}}{NR_1} + \frac{2}{NR_1}\sum_{a \geq 2} x_{a0} + \frac{x_{01}}{NR_2}  \\
& + \frac{2}{NR_2} \Big(N\lambda - x_{10} - \sum_{{a \geq 2}}x_{a0} - x_{01}\Big) \\
\textrm{Subject to:} \; & \frac{x_{10}}{R_1} \leq c_1 \\
\; & \frac{x_{10}}{R_1} + \frac{2}{R_1}\sum_{a \geq 2}x_{a0} \leq c_1 + c_2 \\
\; & \frac{x_{10}}{R_1} + \frac{2}{R_1}\sum_{a \geq 2}x_{a0} +
\frac{x_{01}}{R_2} \leq c_1 + c_2 + c_3
\end{align*}
Letting $y = \sum_{a \geq 2}x_{a0}$ and simplifying the
optimization metric, the above optimization is equivalent to:
\begin{align*}
\textrm{Minimize:} \; & \frac{x_{10}}{N}\Big(\frac{1}{R_1} - \frac{2}{R_2}\Big) + \frac{y}{N}\Big(\frac{2}{R_1} - \frac{2}{R_2}\Big)  \\
& - \frac{x_{01}}{NR_2} + \frac{2\lambda}{R_2} \\
\textrm{Subject to:} \; & \frac{x_{10}}{R_1} \leq c_1 \\
\; & \frac{x_{10}}{R_1} + \frac{2y}{R_1} \leq c_1 + c_2 \\
\; & \frac{x_{10}}{R_1} + \frac{2y}{R_1} + \frac{x_{01}}{R_2}
\leq c_1 + c_2 + c_3
\end{align*}
The above optimization is solved when $x_{10} = R_1c_1$, $x_{10}+2y
= R_1(c_1+c_2)$ and $x_{01} = R_2c_3$. We thus have:
\begin{align}
\overline{e} &\geq \frac{2\lambda}{R_2} + \frac{(c_1 + c_2)}{N} - \frac{R_1}{NR_2}(2c_1 + c_2) - \frac{c_3}{N} \nonumber \\
&= \frac{p+q'}{\theta} + \frac{2}{R_2}\Big(\lambda - \frac{R_1(p+q)+2R_2q'}{2\theta}\Big)
 \label{eq:16}
\end{align}
\end{enumerate}
The necessary set of conditions for $\Phi(\lambda)$ function are
obtained by combining these four bounds. 
\end{proof}

\subsection{Proof of Sufficiency} 
\label{min_energy_algo}

Now we present an algorithm that makes stationary, randomized scheduling 
decisions independent of the actual queue backlog values and show that
for any feasible input rate $\lambda < \mu$, its average energy cost
can be pushed arbitrarily close to the minimum value $\Phi(\lambda)$
with bounded delay. However, the delay bound grows asymptotically as
the average energy is pushed closer to the minimum value.
Similar to the capacity achieving $2$-Hop Relay Algorithm, this algorithm also restricts packets to at most $2$ hops. However, the difference 
lies in that it \emph{greedily} chooses transmission opportunities involving smaller energy cost over other higher cost opportunities. 
An opportunity with higher cost is used \emph{only} when the given input rate cannot be supported using all of the low cost opportunities. 
Thus, depending on the input rate $\lambda$, the algorithm uses only a subset of the transmission opportunities as follows.
\begin{enumerate}
\item If $0 \leq \lambda < \frac{2R_1q}{2\theta}$, all packets are sent using only source-destination transmission opportunities in the same cell.
\item If $\frac{2R_1q}{2\theta} \leq \lambda < \frac{R_1(p+q)}{2\theta}$, all packets are sent either using source-destination transmission 
opportunities in the same cell or source-relay and relay-destination transmission opportunities in the
same cell. 
\item If $\frac{R_1(p+q)}{2\theta} \leq \lambda < \frac{R_1(p+q)+2R_2q'}{2\theta}$, all
packets are sent using same cell transmissions (in either direct transmission or relay modes),
or adjacent cell source-destination transmission opportunities. 
\item And finally, when 
$\frac{R_1(p+q)+2R_2q'}{2\theta} \leq \lambda < \mu$, all transmission
opportunities that restrict packets to at most $2$ hops are used.
\end{enumerate}

To make the presentation simpler, in the following, we only discuss the case
where $\frac{R_1q}{\theta} < \lambda < \frac{R_1(p+q)}{2\theta}$.
The basic idea and performance analysis for the other cases are similar.

Let $\lambda =  \frac{R_1q}{\theta} + \rho \frac{R_1(p-q)}{2\theta}$ where $0 < \rho < 1$ is a given constant.
Also define a control parameter $\beta$ (where $1 < \beta < {1}/{\rho}$) that is input to the algorithm. This parameter affects an energy-delay
tradeoff as shown in Theorem \ref{thm:energy_delay_bound}.

{\bf{Minimum Energy Algorithm:}} Every timeslot, for all
cells, do the following:
\begin{enumerate}
\item If there exists a source-destination pair in the cell,
randomly choose such a pair (uniformly over all such pairs in the
cell). If the source has new packets for the destination, transmit
at rate $R_1$. Else remain idle.

\item If there is no source-destination pair in the cell but there
are at least $2$ users in the cell, then with probability $\beta \rho$, decide to use
the same cell relay transmission opportunity as described in the next step. Else remain idle.

\item If decide to use the same cell relay transmission opportunity in step $(2)$,
randomly designate one user as the sender and another as the receiver. 
Then with probability $\frac{1-\delta}{2}$ (where $0 < \delta < 1$ and $\delta = \delta(\beta)$ to be determined later) 
perform the first action below. Else, perform the second.
\begin{enumerate}
\item \emph{Send new Relay packets in same cell}: If the transmitter
has new packets for its destination, transmit at rate $R_1$. Else
remain idle.
\item \emph{Send Relay packets to their Destination in same
cell}: If the transmitter has packets for the receiver, transmit at
rate $R_1$. Else remain idle.
\end{enumerate}
\end{enumerate}
Note that the above algorithm does not use any adjacent cell
transmission opportunities. All packets are sent over at most $2$
hops using only same cell transmissions. We now analyze the
performance of this algorithm.
\begin{thm} 
For the cell partitioned network (with $N$ nodes and $C$ cells) as described in Sec. \ref{section:model}, 
with minimum energy function $\Phi(\lambda)$ as described above, 
and user mobility model as described in Sec. \ref{section:mob_model},
the average energy cost $\overline{e}$ of the Minimum Energy Algorithm with input
rate $\lambda$ for each user such that $\lambda = \frac{R_1q}{\theta} + \rho \frac{R_1(p-q)}{2\theta}$ (where $0 < \rho < 1$), 
a control parameter $\beta$ (where $ 1 < \beta < {1}/{\rho}$), and with $\delta = \frac{\beta -1}{2 \beta}$ satisfies:
\begin{align}
\overline{e}
 = \Phi(\lambda) + (\beta-1) \rho \Big(\frac{p-q}{\theta}\Big) 
\label{eq:energy_bound}
\end{align}
and the average packet delay $\overline{D}$ satisfies:
\begin{align}
\overline{D} \leq \frac{4BN\theta (2d+1)}{\lambda R_1 (p-q) \rho (\beta-1)}
\label{eq:delay_bound3}
\end{align}
\label{thm:energy_delay_bound}
\end{thm}
where ${B}$ is a constant given by (\ref{eq:B}) 
and $d$ is a finite integer that is related to the
mixing time of the joint user mobility process and is given by 
$d = \Big\lceil{\frac{\log\Big( \frac{4 N^2(p+q)\alpha\beta}{(p-q)\rho(\beta-1)}\Big)}{\log(1/\gamma)}}\Big\rceil$. 

From the above, it can be seen that the control parameter $\beta$ allows a $(O(\beta -1), O(1/(\beta - 1)))$ 
tradeoff between the average energy cost and the average delay bound. Specifically, the average energy cost
$\overline{e}$ can be pushed arbitrarily close to $\Phi(\lambda)$ by pushing $\beta$ closer to $1$. However, this increases
the bound on $\overline{D}$ as ${1}/(\beta-1)$.

\begin{proof} The proof is similar to the proof of Theorem \ref{thm:delay_bound} and is
given in Appendix C. 
\end{proof}


\section{Capacity Gains by Network Coding}
\label{section:nc_gains}

Here, we show an example 
where the network capacity
can be strictly improved by making use of network coding in conjunction with the wireless broadcast advantage. Specifically,
consider a network with $6$ nodes and $4$ cells. Suppose the steady-state location distribution for all nodes is uniform over all cells.
Thus, $\pi_c = 1/4$ for all $c$.
The one-to-one traffic pairing is given by $1 \leftrightarrow 2, 3 \leftrightarrow 4, 
5 \leftrightarrow 6$. Let $R_1 = 1$ and $R_2 = 0$. Thus, this example only allows same cell transmissions.
We further assume that when a node in a cell transmits, all other nodes in that cell receive that packet. Note that
the $2$-Hop Relay Algorithm presented in Sec. \ref{section:2hop_relay} does not make use of this feature. 

Using Theorem \ref{thm:capacity}, the network capacity under the model presented in Sec. \ref{section:model} can be computed. Specifically, the network
capacity is given by $\mu = \frac{q + p}{2\theta}$ packets/slot per node where $\theta = \frac{6}{4}$ and using (\ref{eq:prob_values}),
we have $q = 1 - \Big(1 - \frac{1}{16}\Big)^3$ and $p = 1 - \Big(1 - \frac{1}{4}\Big)^6 - \frac{6}{4}\Big(1 - \frac{1}{4}\Big)^5$.

  We now show how network coding can be used to achieve a throughput that is strictly higher than $\mu$. First we define $4$ distinct
configurations of the nodes. In configuration $\mathrm{I}$, nodes $1, 4,$ and $5$ are in the same cell and the other nodes can be in any of the remaining cells (but not in
the same cell as nodes $1, 4,$ and $5$). Note that this cell can be any one of the $4$ cells. From the assumption about the node mobility process, the
steady-state probability of configuration $\mathrm{I}$ is given by 
$\nu \defequiv 4\times (\frac{1}{4})^3\times(\frac{3}{4})^3$. 
In configuration $\mathrm{II}$, nodes $2, 3,$ and $5$ are in the same cell and the other nodes can be in any of the remaining cells (but not in
the same cell as nodes $2, 3,$ and $5$). In configuration $\mathrm{III}$, nodes $2, 4,$ and $5$ are in the same cell and the other nodes 
can be in any of the remaining cells (but not in the same cell as nodes $2, 4,$ and $5$). Finally, in 
configuration $\mathrm{IV}$, nodes $1, 3,$ and $5$ are in the same cell and the other nodes could be in any of the remaining cells (but not in
the same cell as nodes $1, 3,$ and $5$). Note that these configurations cannot occur simultaneously as each consists of node $5$. Further,
the steady-state probability of each configuration is given by $\nu$.

In the following, we will modify the $2$-Hop Relay Algorithm of Sec. \ref{section:2hop_relay} when one of these configurations
occur in any cell and demonstrate an improvement in the throughput of nodes $1, 2, 3$ and $4$ over $\mu$. For each configuration,
we will only focus on the transmissions in the cell with the three nodes that define that configuration. The 
$2$-Hop Relay Algorithm for the other cells remains the same.

Note that under each configuration, there are no source-destination pairs in the cell of interest. Thus,
 under the $2$-Hop Relay Algorithm, a node is selected as the transmitter with probability $\frac{1}{3}$ while the remaining two nodes are equally
likely to be selected as the receiver. Further, the transmitter is scheduled to
transmit a new packet to the receiver with probability $\frac{1-\delta}{2}$ and is scheduled to
transmit a relay packet to the receiver with probability $\frac{1+\delta}{2}$. 
Thus, in each configuration, each of the two nodes other than node $5$
is scheduled to transmit a new packet to node $5$ with probability $\frac{1}{3} \times \frac{1}{2} \times \frac{(1-\delta)}{2} = \frac{(1-\delta)}{12}$.
Also, in each configuration, node $5$ is scheduled to transmit a relay packet to each of the other two nodes in the cell with probability 
$\frac{1}{3} \times \frac{1}{2} \times \frac{(1+\delta)}{2} = \frac{(1+\delta)}{12}$. Adding the probabilities associated with these four scheduling decisions yields
\begin{align}
\frac{(1-\delta)}{12} + \frac{(1-\delta)}{12} + \frac{(1+\delta)}{12} + \frac{(1+\delta)}{12} = \frac{1}{3}
\label{eq:sum_probs}
\end{align}

We now modify the $2$-Hop Relay Algorithm to take advantage of network coding. For all configurations other than the four as defined above,
the algorithm remains the same. However, in each of the configurations $\mathrm{I, II, III, IV}$, we change the probability of scheduling a node to
transmit a new packet (for relaying) to node $5$ from 
$\frac{(1-\delta)}{12}$
to $\frac{1}{3} \times \frac{(1-\epsilon)}{3} = \frac{(1-\epsilon)}{9}$ 
where $0 < \epsilon < 1$.
Also, node $5$ is scheduled to transmit a relay packet to the other two nodes in the cell with probability   
$\frac{1}{3} \times \frac{(1+2\epsilon)}{3} = \frac{(1+2\epsilon)}{9}$. However, whenever node $5$ has at least one packet for each of the two other nodes,
it broadcasts a XOR of two packets destined for these nodes in a single transmission. If node $5$ does not have at least one packet for each of the two other 
nodes, it would simply transmit a regular packet (if available).
Note that under the original $2$-Hop Relay Algorithm, the two scheduling decisions of node $5$ transmitting a relay packet to 
the other two nodes are taken with probability $\frac{(1+\delta)}{12}$ each. 
These are now replaced by a \emph{single} scheduling decision of node $5$ broadcasting a XORed relay packet and this has probability $\frac{(1+2\epsilon)}{9}$.
The probabilities associated with the other scheduling decisions under this modified algorithm remain the same as the original
$2$-Hop Relay Algorithm. 
The sum of probabilities associated with the modified scheduling decisions as described above is given by
\begin{align}
\frac{(1-\epsilon)}{9} + \frac{(1-\epsilon)}{9} + \frac{(1+2\epsilon)}{9} = \frac{1}{3}
\label{eq:sum_probs_mod}
\end{align}
This is the same as (\ref{eq:sum_probs}).
Thus, it can be seen that the probabilities of all scheduling decisions under the modified algorithm sum to $1$.

\begin{figure}
\centering
\includegraphics[width=7.5cm, height = 7.6cm, angle=0]{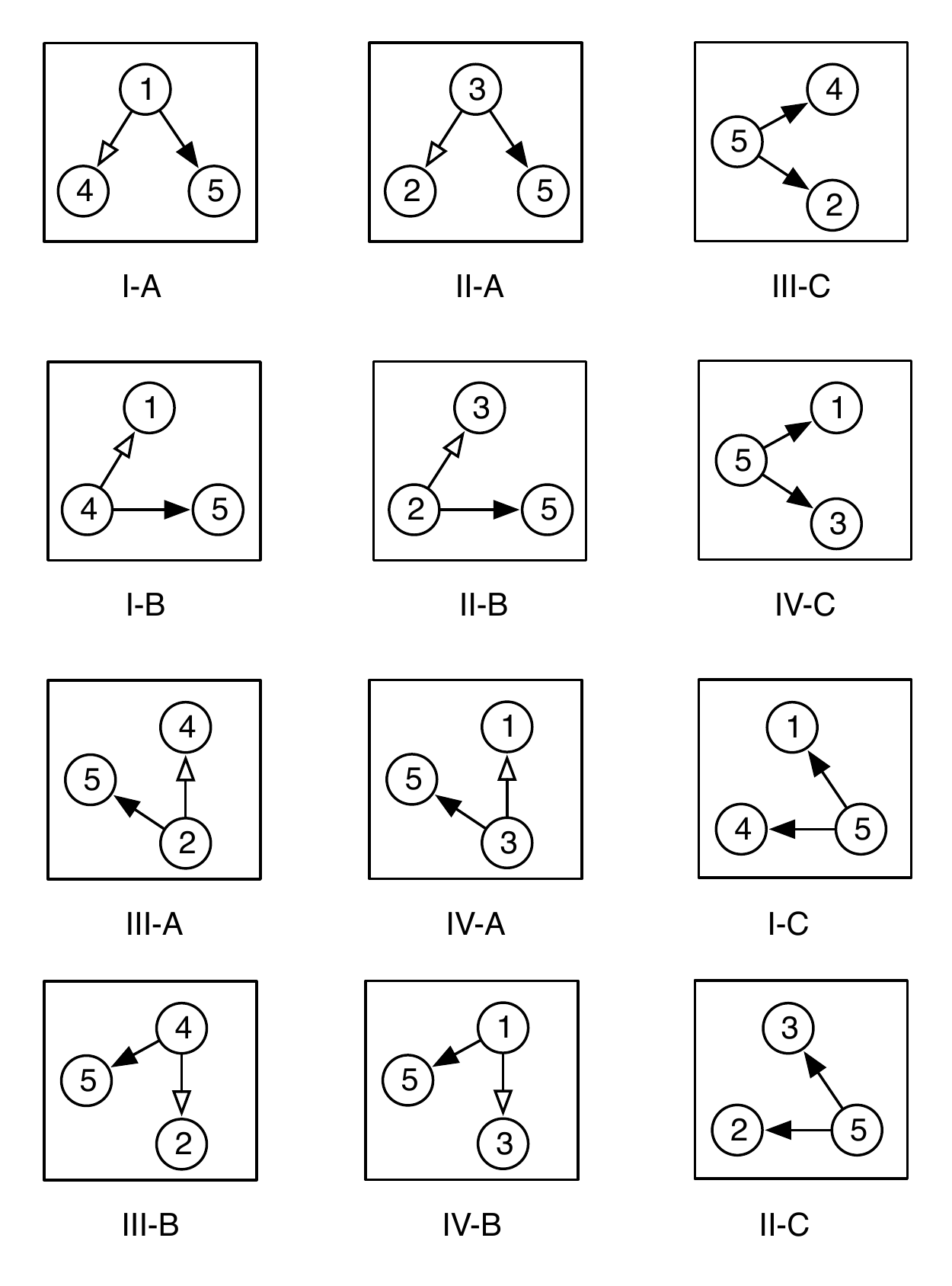}
\caption{An example showing capacity gains possible by using network coding 
in conjunction with the wireless broadcast advantage.}
\label{fig:3}
\end{figure}

   To see how the nodes can recover the original packets from the XORed packet, we further classify each configuration into type A, B and C 
depending on the scheduling decision as shown in Fig. \ref{fig:3}. The configurations of type A and B correspond to the scheduling decisions in which
a node is scheduled to transmit a new packet (for relaying) to node $5$. The configurations of type C correspond to the scheduling decisions in
which node $5$ is scheduled to transmit relay packets to the other two nodes (either as a network coded XOR packet whenever possible or a regular packet).
In each configuration of type A or B, whenever a new packet is transmitted by a node to node $5$ for relaying, the other node overhears the packet
and stores a copy. For example, in Fig. $\mathrm{I}$-A, when node $1$ transmits a new packet (destined for node $2$) to node $5$, 
node $4$ overhears this transmission and stores a copy of this packet. Similarly, in $\mathrm{II}$-A, when node $3$ transmits a new packet 
(destined for node $4$) to node $5$, node $2$ overhears this transmission and stores a copy of
this packet. In each configuration of type C, whenever node $5$ has at least one packet for each of the two other nodes, note that
each of these two nodes already has a copy of the packet destined for the other node (that it obtained by overhearing earlier in a type A or B configuration).
Therefore, when node $5$ transmits a XOR packet, both of these nodes can recover the original packets destined for them
by using the side information already available to them in the form of previously overheard and stored packets. 
For example, in $\mathrm{III}$-C, node $5$ is in the same cell as nodes $2$ and $4$ and suppose it has a packet for each of them.
Then, node $2$ must have the packet that is destined for node $4$ that it overheard in $\mathrm{II}$-A. Similarly,
node $4$ must have the packet that is destined for node $2$ that it overheard in $\mathrm{I}$-A.
Thus, when node $5$ broadcasts a XOR packet in a single transmission, both nodes $2$ and $4$ can retrieve their desired packets. 
Thus, this single transmission effectively 
delivers two packets. Note that under a scheme that does not allow mixing of packets, at most one packet can be transmitted per transmission.

To demonstrate gains in throughput under this ``network coding enhanced'' $2$-Hop Relay Algorithm, we define the following additional relay queues
at node $5$ as shown in Fig. \ref{fig:4}. Arrivals to and departures from these queues happen only when scheduling decisions corresponding to the $12$ 
configurations in Fig. \ref{fig:3} are made according to the enhanced $2$-Hop Relay Algorithm. 
$U_{ij}^{(i)}(t)$ and $U_{ij}^{(j)}(t)$ refer to the queue of packets destined for nodes $i$ and $j$
respectively that will be network coded whenever possible. Fig. \ref{fig:4} shows the arrival rates and the corresponding configurations
(when arrivals happen to these queues) as well as the service rates and corresponding configurations (when packets are served from these queues). 
Note that each queue has an arrival rate of $\frac{(1-\epsilon)\nu}{9}$ and
sees a service rate of $\frac{(1+2\epsilon)\nu}{9}$. 
Since  $(1+2\epsilon) > (1-\epsilon)$, all these
queues are stable. The additional throughput for nodes $1, 2, 3$ and $4$ over the $2$-Hop Relay Algorithm without network coding
is given by $\Big[\frac{2(1 - \epsilon)}{9} - \frac{2(1-\delta)}{12}\Big]\nu$ packets/slot. 
This is strictly positive for any $0 < \epsilon < \frac{1}{4}$. For example, by choosing $\epsilon = \frac{1}{8}$,
a throughput gain of $\frac{\nu}{36}$ packets/slot is achievable. Thus, the capacity can be strictly increased over a scheme that is restricted to pure routing.

\begin{figure}
\centering
\includegraphics[width=7.5cm,height = 7.6cm, angle=0]{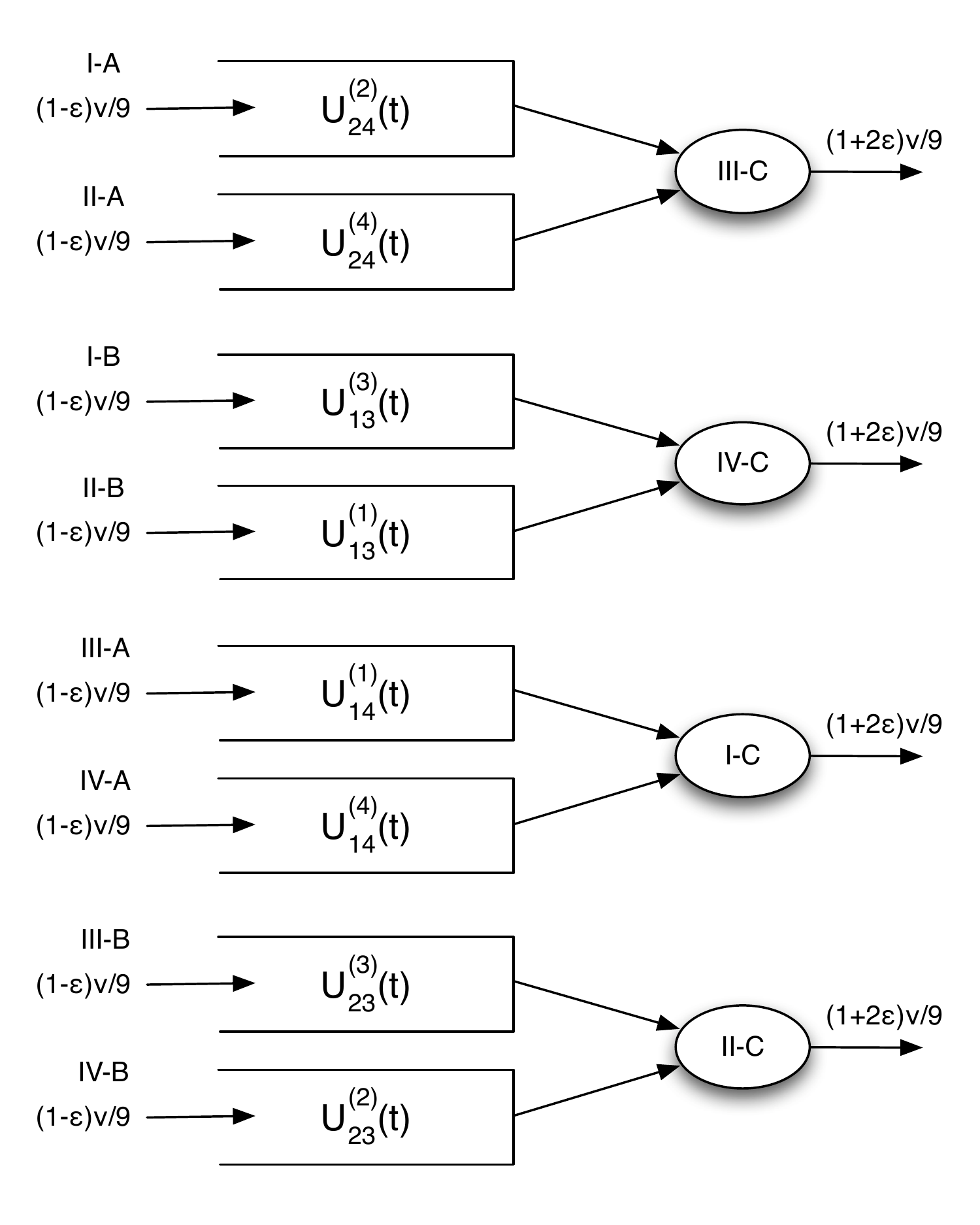}
\caption{Additional relay queues at node $5$ under the network coding enhanced $2$-Hop Relay Algorithm that is used 
in configurations $\mathrm{I, II, III, IV}$.} 
\label{fig:4}
\end{figure}

\section{Conclusions}
\label{section:con}

In this work, we investigated two quantities of fundamental interest
in a delay-tolerant mobile ad hoc network: the network capacity and the minimum
energy function. Using a cell-partitioned model of the network, we
obtained \emph{exact} expressions for both these quantities in terms
of the network parameters (number of nodes $N$ and number of cells
$C$) and the steady-state location distribution of the mobility
process. Our results hold for general mobility processes (possibly
non-uniform and non-i.i.d.) and our analytical technique can be
extended to other models with additional scheduling constraints.

We also proposed two simple scheduling strategies that can achieve
these bounds arbitrarily closely at the cost of an increased delay.
Both these schemes restrict packets to at most $2$ hops and
make scheduling decisions purely based on the current
user locations and independent of the actual queue backlogs. 
For both schemes, we computed bounds on the average packet delay
using a Lyapunov drift technique. 

In this paper, we have focused on network control algorithms that operate according
to the  network structure as presented in Sec. \ref{section:model}. We assumed that the packets
themselves are kept intact and are not combined or network coded. 
As shown in the example in Sec. \ref{section:nc_gains}, it is possible to increase the
network capacity by making use of network coding and the wireless
broadcast feature. An
interesting future direction of this research is to determine
the exact capacity region with such enhanced control options.

\section*{Appendix A \\ Proof of Lemma $1$}

Here, we prove the bound in Lemma \ref{lem:lem1}. 
We have
\begin{align*}
&\expect{f(\vec{\chi}(t+d))| \vec{\chi}(t)} = \sum_{\vec{c}} f(\vec{c}) \times Pr\{\vec{\chi}(t+d) = \vec{c} | \vec{\chi}(t)\} \\
&= \sum_{c_1, c_2, \ldots, c_N} f(c_1, c_2, \ldots, c_N) \Big[ \prod_{i=1}^N Pr\{\chi_i(t+d) = c_i | \vec{\chi}(t)\} \Big] \\
&\geq \sum_{c_1, c_2, \ldots, c_N} f(c_1, c_2, \ldots, c_N)\Big[ \prod_{i=1}^N \pi_{c_i} (1 - \alpha \gamma^d) \Big] \\
&= \sum_{c_1, c_2, \ldots, c_N} f(c_1, c_2, \ldots, c_N) \Big(\prod_{i=1}^N \pi_{c_i} \Big) (1 - \alpha \gamma^d)^N \\
&\geq \sum_{c_1, c_2, \ldots, c_N} f(c_1, c_2, \ldots, c_N)  \Big(\prod_{i=1}^N \pi_{c_i} \Big) (1 - 2N\alpha \gamma^d) \\
& = f_{av} (1 - 2N\alpha \gamma^d)   
\end{align*}
where step two follows from the independence of node mobility processes, step three follows from (\ref{eq:renewal}) and the
 second last step uses the inequality $(1 - \alpha \gamma^d)^N \geq (1 - 2N\alpha \gamma^d)$. This can be shown by induction as follows.
This holds for $N=1$. Suppose it holds for some integer $i > 1$, i.e., $(1 - \alpha \gamma^d)^i \geq (1 - 2i\alpha \gamma^d)$. Then, we have
\begin{align*}
(1 - \alpha \gamma^d)^{i+1} = (1 - \alpha \gamma^d)^{i}(1 - \alpha \gamma^d) &\geq (1 - 2i\alpha \gamma^d)(1 - \alpha \gamma^d) \\
&\geq (1 - 2(i+1)\alpha \gamma^d)
\end{align*}
The upper bound can be shown similarly, except that we use the inequality 
$(1 + \alpha \gamma^d)^N \leq (1 + 2N\alpha \gamma^d)$ for all $N \geq 2$ whenever
$d$ is such that $\alpha \gamma^d < 1/N^2$. To show this, 
let $\alpha \gamma^d = c/N^2$ where $0 < c < 1$. Note that $0 < c/N < 1$. Then 
\begin{align*}
&(1 + \alpha \gamma^d)^{N} \\
&= 1 + N \alpha \gamma^d + \frac{N(N-1)}{2} (\alpha \gamma^d)^2 + \ldots + (\alpha \gamma^d)^N \\
&< 1 + \frac{c}{N} + \Big(\frac{c}{N}\Big)^2 + \ldots + \Big(\frac{c}{N}\Big)^N < \frac{1}{1 - \frac{c}{N}} = \frac{N}{N-c}\\
&< 1 + \frac{c}{N-1} = 1 + \frac{N^2 \alpha \gamma^d}{N-1}\leq 1 + 2N \alpha \gamma^d \qquad \forall N \geq 2
\end{align*}

\section*{Appendix B \\ Derivation of Probability Expressions}

In what follows, we will use the linearity of expectations property
to compute the probability expressions in
(\ref{eq:prob_values}).

 \emph{Derivation of $q$:} Let $I_c(t)$ be an indicator variable that is
$1$ if there is a source-destination pair in cell $c$ in slot $t$ in
the steady-state. Then the expected number of cells with a
source-destination pair is given by $\expect{\sum_{c=1}^C I_c(t)}$.
By linearity of expectations, this is equal to $\sum_{c=1}^C
\expect{I_c(t)}$. To compute $\expect{I_c(t)}$ for any cell $c$,
note that by the independence of user mobility processes, $\pi_c^2$
is the probability of finding any particular source-destination pair
in cell $c$ in the steady-state. Since there are $N/2$ such pairs and they occur independent of each other,
the probability of finding no source-destination pair in cell $c$ in
the steady-state is given by $(1-\pi_c^2)^{N/2}$. Thus, the
probability of finding \emph{at least} $1$ source-destination pair
is $1 - (1-\pi_c^2)^{N/2}$. Using this, we get $q =
\frac{1}{C}\sum_{c=1}^C (1 - (1-\pi_c^2)^{N/2})$.

 \emph{Derivation of $p$:} To compute the probability of finding at least
$2$ users in a cell $c$, we note that this can be obtained by first
computing the probabilities of finding no user and exactly $1$ user
in cell $c$ and then subtracting these from $1$. These are given by
$(1-\pi_c)^{N}$ and $\binom{N}{1} \pi_c(1-\pi_c)^{N-1}$
respectively. Using this, we get $p = \frac{1}{C}\sum_{c=1}^C (1 -
(1-\pi_c)^{N} - N\pi_c(1-\pi_c)^{N-1})$.

 \emph{Derivation of $q'$:} The probability of finding exactly $1$ user
in cell $c$ is given by $\binom{N}{1}\pi_c (1-\pi_c)^{N-1}$. The
probability of finding its destination in an adjacent cell
\emph{given that} it is not it cell $c$ is given by $\frac{1}{1 -
\pi_c}\sum_{i\in \mathcal{B}_c} \pi_i$ which we have defined as
$\Pi_{adj}(c)$. Using this, we get $q' = \frac{1}{C}\sum_{c=1}^C
(\Pi_{adj}(c)N\pi_c (1-\pi_c)^{N-1})$.

 \emph{Derivation of $p'$:} Given that there is exactly $1$ user in cell
$c$, the probability that at least $1$ of the remaining $N-1$ users
is in an adjacent cell is given by $1 - (1 - \Pi_{adj}(c))^{N-1}$.
Thus, we get $p' = \frac{1}{C}\sum_{c=1}^C (1 - (1 -
\Pi_{adj}(c))^{N-1})N\pi_c (1 - \pi_c)^{N-1}$.

 \emph{Derivation of $q''$:} We first compute the probability of finding
$i$ users in cell $c$ such that there are no source-destination
pairs. Clearly, $1 \leq i \leq \frac{N}{2}$ since there must be at
least $1$ source-destination pair for $i > \frac{N}{2}$. Next, note
that $2^i \frac{\binom{N/2}{i}}{\binom{N}{i}}$ is the probability of
finding no source-destination pair in a cell \emph{given that} there
are $i$ users in that cell. $\binom{N}{i}\pi_c^i (1 -
\pi_c)^{N-i}$ is the probability of having $i$ users in cell $c$.
Finally, the probability that there is at least $1$ node in an
adjacent cell that will make a source-destination pair with one of
these $i$ nodes \emph{given that} it is not in cell $c$ is given by
$(1 - (1 - \Pi_{adj}(c))^{i})$. Combining all these, we get 
\begin{align*}
q'' = \frac{1}{C} \sum_{c=1}^C \sum_{i=1}^{N/2} 2^i\binom{N/2}{i} \pi_c^i
(1 - \pi_c)^{N-i} (1 - (1 - \Pi_{adj}(c))^{i})
\end{align*}

 \emph{Derivation of $p''$:} Similar to the derivation of $q''$, the
probability of finding $i$ users in cell $c$ such that there are no
source-destination pairs in cell $c$ as well as any adjacent cells
is given by $2^i \binom{N/2}{i} \pi_c^i (1 - \pi_c)^{N-i} (1 -
\Pi_{adj}(c))^{i}$. Since we also want at least $2$ users in cell
$c$, we sum from $i=2$ to $N/2$. This yields 
\begin{align*}
p'' = \frac{1}{C} \sum_{c=1}^C \sum_{i=2}^{N/2} 2^i \binom{N/2}{i} \pi_c^i (1 - \pi_c)^{N-i} (1 - \Pi_{adj}(c))^{i}
\end{align*}

\section*{Appendix C \\ Proof of Theorem \ref{thm:energy_delay_bound}}

Here, we establish the bounds (\ref{eq:energy_bound}) and (\ref{eq:delay_bound3}).

When $\frac{R_1q}{\theta} < \lambda < \frac{R_1(p+q)}{2\theta}$,
under the Minimum Energy Algorithm, all transmissions are either same cell direct transmissions or same cell
relay transmissions. Specifically,  each user either transmits directly to its destination or transmits
new packets to a relay or transmits relayed packets to their
destination in the same cell. 
Each such transmission involves one unit of energy cost and
therefore the average energy cost per user $\overline{e}$ can be expressed in terms of
the rates of these transmission opportunities. The rate at which same cell direct transmissions
are scheduled is given by $C{q}$. The rate at which same cell relay transmissions
are scheduled is given by $\beta\rho C{(p-q)}$. 
Thus, we have:
\begin{align*}
\overline{e}  &= \frac{1}{N} \Big[C{q} + \beta\rho C(p-q)\Big] = \frac{q}{\theta} + \rho \frac{p-q}{\theta} + (\beta -1)\rho \frac{p-q}{\theta} \\
&= \Phi(\lambda) + (\beta-1)\rho \Big(\frac{p-q}{\theta}\Big)
\end{align*}
Thus, $\overline{e}$ can be pushed arbitrarily close to $\Phi(\lambda)$ by choosing $\beta$ close to $1$.

The delay of the Minimum Energy Algorithm can be analyzed using a procedure similar
to the one used in the proof of Theorem \ref{thm:delay_bound}. 
We first evaluate bounds on the expression in (\ref{eq:delayed}) by computing
the steady-state service rates $\overline{\mu}_{ib}^{(c)}, \overline{\mu}_{ai}^{(c)}$ 
achieved by the Minimum Energy Algorithm. We have the following $2$ cases:

$1)$ Node $i$ generates type $c$ packets: In this case,
$\expect{A_{i}^{(c)}(t)} = \lambda$ and  $\sum_a {\mu}_{ai}^{(c)}(t) = 0$.
To calculate $\sum_b \overline{\mu}_{ib}^{(c)}$, 
define $r_1, r_2, r_3$ similar to that in the proof of Theorem \ref{thm:delay_bound}.
Then, the total rate of transmission over the
network is given by $N(r_1 + r_2 + r_3)$. 
Similar to Theorem \ref{thm:delay_bound}, we have
$r_2 = \frac{1-\delta}{1 + \delta}r_3$. Since only same cell
direct transmissions are used, we have $Nr_1 = CR_1q$.
Also, a same cell relay transmission is scheduled with probability $\beta \rho$ 
whenever there is no source-destination pair in the cell but there
are at least $2$ users in the cell,
Thus, the sum total transmissions in the network can be
expressed in terms of the quantities $p$ and $q$ as
$N(r_1 + r_2 + r_3) = C(R_1q + R_1\beta\rho (p-q))$.
Solving for $r_1, r_2, r_3$, we have: 
\begin{align}
&r_1 = \frac{R_1 q}{\theta}, r_2 = \frac{R_1 (p-q)(1-\delta)\beta \rho}{2\theta}\nonumber \\&
 r_3 = \frac{R_1 (p-q) (1+\delta)\beta \rho }{2\theta}
\label{eq:r3_e3}
\end{align}
Therefore, we have:
\begin{align*}
\sum_b \overline{\mu}_{ib}^{(c)} = r_1 + r_2 = \frac{R_1 q}{\theta} + \frac{R_1 (p-q)(1-\delta)\beta \rho}{2\theta}
\end{align*}
Let $\delta = \frac{\beta-1}{2\beta}$ and  $\alpha \gamma^d = \frac{(p-q)\rho\delta}{2(p+q)N^2} = \frac{(p-q)\rho(\beta-1)}{4\beta(p+q)N^2} < \frac{1}{N^2}$.
Note that this choice of $\delta$ can be shown to represent a valid probability,
because $1 < \beta < \frac{1}{\rho} \Rightarrow 0 < \frac{1}{2} - \frac{1}{2 \beta} < \frac{1}{2} - \frac{\rho}{2} \Rightarrow 0 < \delta < 1$.
Then, using (\ref{eq:stat6-1}), the last term of (\ref{eq:delayed}) under this case can be expressed as:
\begin{align*}
&\expect{\sum_b \mu_{ib}^{(c)}(t) - \sum_a \mu_{ai}^{(c)}(t) - A_i^{(c)}(t) |\vec{U}(t-d)} \geq \\
&(r_1 + r_2)(1 -  2N\alpha \gamma^d ) - \lambda \geq (r_1 + r_2) - \frac{R_1(p-q)\rho\delta}{2 \theta N} - \lambda \\
&= \frac{R_1 (p-q) \rho}{2 \theta} \Big[(1 - \delta)\beta - \frac{\delta}{N} - 1 \Big] \geq \frac{R_1 (p-q) \rho (\beta-1)}{8 \theta} 
\end{align*}
where we used the relations $\lambda = r_1 + \frac{r_2}{(1-\delta)\beta}$,
$(r_1 + r_2)2N\alpha \gamma^d \leq \frac{R_1(p-q)\rho\delta}{2 \theta N}$ and 
$(1 - \delta)\beta - \frac{\delta}{N} - 1 \geq \frac{\beta -1}{4}$.
These can be shown as follows.
Using (\ref{eq:r3_e3}), we have $\lambda = \frac{R_1q}{\theta} + \rho \frac{R_1(p-q)}{2\theta} = r_1 + \frac{r_2}{(1-\delta)\beta}$.
Next:
\begin{align*}
&(r_1 + r_2)2N\alpha \gamma^d = \Big(\frac{R_1 q}{\theta} + \frac{R_1 (p-q)(1-\delta)\beta \rho}{2\theta}\Big) \frac{(p-q)\rho\delta}{(p+q)N}\\
&< \Big(\frac{R_1 q}{\theta} + \frac{R_1 (p-q)}{2\theta}\Big) \frac{(p-q)\rho\delta}{(p+q)N} \qquad (\textrm{since $(1-\delta)\beta\rho < 1$}) \\
&= \frac{R_1 (p-q)\rho \delta}{2\theta N}
\end{align*}
Finally, using $\delta = \frac{\beta - 1}{2\beta}$, we have:
\begin{align*}
&(1 - \delta)\beta - \frac{\delta}{N} - 1 = \frac{\beta + 1}{2} - \frac{\delta}{N} - 1 \geq \frac{\beta - 1}{2} - \frac{\delta}{2} \\
&= \frac{\beta - 1}{2} - \frac{\beta - 1}{4\beta} \geq \frac{\beta -1}{4}
\end{align*}


$2)$ Node $i$ relays type $c$ packets: From our traffic model, we know that in this case
$A_{i}^{(c)}(t) = 0$ for all $t$. 
To compute $\sum_b \overline{\mu}_{ib}^{(c)}$ and $\sum_a \overline{\mu}_{ai}^{(c)}$, note
that the Minimum Energy Algorithm schedules relay transmissions such that all $N-2$ relay packet types are
equally likely. Thus we have:
\begin{align*}
\sum_b \overline{\mu}_{ib}^{(c)} = \frac{r_3}{N-2}, \;\; \sum_a \overline{\mu}_{ai}^{(c)} = \frac{r_2}{N-2}
\end{align*}
Let $\delta = \frac{\beta-1}{2\beta}$ and  $\alpha \gamma^d = \frac{(p-q)\rho\delta}{2(p+q)N^2} = \frac{(p-q)\rho(\beta-1)}{4\beta(p+q)N^2} < \frac{1}{N^2}$.
Then, using (\ref{eq:stat6-1}), (\ref{eq:stat6-2}), the last term of (\ref{eq:delayed}) under this case can be expressed as:
\begin{align*}
&\expect{\sum_b \mu_{ib}^{(c)}(t) - \sum_a \mu_{ai}^{(c)}(t) - A_i^{(c)}(t) |\vec{U}(t-d)} \\
&\geq \Big( \sum_b \overline{\mu}_{ib}^{(c)} \Big) (1 -  2N\alpha \gamma^d ) - \Big( \sum_a \overline{\mu}_{ai}^{(c)} \Big) (1 +  2N\alpha \gamma^d ) \\
&= \Big( \frac{r_3 - r_2}{N-2} \Big) - \Big( \frac{r_3 + r_2}{N-2} \Big)2N\alpha \gamma^d \\
&\geq \frac{R_1 (p-q)\rho \beta }{\theta N} \Big[ \delta - \frac{\delta}{N} \Big] \geq \frac{R_1 (p-q) \rho (\beta-1)}{4N\theta}
\end{align*}
where we used the inequality $2N\alpha \gamma^d < \frac{\delta}{N}$.
Combining these two cases, with $\delta = \frac{\beta-1}{2\beta}$ and 
$\alpha \gamma^d = \frac{(p-q)\rho(\beta-1)}{4\beta(p+q)N^2}$
we have $\expect{\sum_b \mu_{ib}^{(c)}(t) - \sum_a \mu_{ai}^{(c)}(t) - A_{i}^{(c)}(t)|\vec{U}(t-d)} \geq \frac{R_1 (p-q) \rho (\beta-1)}{4N\theta}$.
Using this in (\ref{eq:delayed}), we get:
\begin{align*}
\Delta(t, d) \leq & BN(2d+1) \nonumber \\  
&- \frac{R_1 (p-q) \rho (\beta-1)}{4N\theta} \sum_{i=1}^N\sum_{c \neq i} {U^{(c)}_i(t-d)} 
\end{align*}
This is in a form that fits (\ref{eq:ldlemma}). Using the Lyapunov
Drift Lemma, we get

\begin{align*}
\limsup_{t\rightarrow\infty} \frac{1}{t}\sum_{\tau=0}^{t-1}\sum_{i
\neq c}\expect{U_i^{(c)}(\tau) } \leq \frac{4BN^2\theta (2d+1)}{ R_1 (p-q) \rho (\beta-1)}
\end{align*}
The total input rate into the network is $N\lambda$. Thus,  using
Little's Theorem, the average delay per packet is bounded by
$\frac{4BN\theta (2d+1)}{\lambda R_1 (p-q) \rho (\beta-1)}$.

\end{document}